\documentclass[aps,prx,twocolumn,english,superscriptaddress,floatfix]{revtex4-2}
\usepackage{times,color,amsfonts,amssymb,amsbsy,amsthm,amsmath,graphics,graphicx,hyperref,bm,bbm,makecell,mathtools,dsfont,upgreek,hyperref,fancyhdr,multirow}
\usepackage{blindtext}
\usepackage[capitalise]{cleveref}
\usepackage{braket}
\def\diracspacing{0.7pt}
\newcommand{\ketbra}[2]{|\hspace{\diracspacing} #1 \rangle \langle #2 \hspace{\diracspacing} |}

\DeclareMathOperator{\Tr}{Tr}

\newtheorem{Lmm}{Lemma}
\newtheorem{Thm}{Theorem}

\usepackage{tcolorbox} 
\hypersetup{colorlinks,linkcolor={blue},citecolor={blue},urlcolor={blue}}
\begin{document}
\title{Satellite-Based Quantum Key Distribution in the Presence of Bypass Channels}
\author{Masoud Ghalaii}
\affiliation{School of Electronic and Electrical Engineering, University of Leeds, Leeds LS2 9JT, United Kingdom}
\affiliation{Department of Computer Science, University of York, York YO10 5GH, United Kingdom}
\author{Sima Bahrani}
\affiliation{Department of Electrical and Electronic Engineering, University of Bristol, Bristol BS8 1UB, United Kingdom}
\author{Carlo Liorni}
\affiliation{Institute f{\"u}r Theoretische Physik III, Heinrich Heine Universit{\"a}t, D-40225 D{\"u}sseldorf, Germany}
\author{Federico Grasselli}
\affiliation{Institute f{\"u}r Theoretische Physik III, Heinrich Heine Universit{\"a}t, D-40225 D{\"u}sseldorf, Germany}
\author{Hermann Kampermann}
\affiliation{Institute f{\"u}r Theoretische Physik III, Heinrich Heine Universit{\"a}t, D-40225 D{\"u}sseldorf, Germany}
\author{Lewis Wooltorton}
\affiliation{Department of Electrical and Electronic Engineering, University of Bristol, Bristol BS8 1UB, United Kingdom}
\affiliation{Quantum Engineering Centre for Doctoral Training, H. H. Wills Physics Laboratory, University of Bristol, Bristol BS8 1FD, United Kingdom}
\affiliation{Department of Mathematics, University of York, Heslington, York, YO10 5DD, United Kingdom}
\author{Rupesh Kumar}
\affiliation{Department of Physics, University of York, York YO10 5DD, United Kingdom}
\affiliation{York Centre for Quantum Technologies, University of York, York, United Kingdom}
\author{Stefano Pirandola}
\affiliation{Department of Computer Science, University of York, York YO10 5GH, United Kingdom}
\author{Timothy P. Spiller}
\affiliation{Department of Physics, University of York, York YO10 5DD, United Kingdom}
\affiliation{York Centre for Quantum Technologies, University of York, York, United Kingdom}
\author{Alexander Ling}
\affiliation{Centre for Quantum Technologies, National University of Singapore, Singapore}
\author{Bruno Huttner}
\affiliation{ID Quantique, Geneva, Switzerland}
\author{Mohsen Razavi}
\affiliation{School of Electronic and Electrical Engineering, University of Leeds, Leeds LS2 9JT, United Kingdom}

\begin{abstract} 
The security of prepare-and-measure satellite-based quantum key distribution (QKD), under restricted eavesdropping scenarios, is addressed. We particularly consider cases where the eavesdropper, Eve, has limited access to the transmitted signal by Alice, and/or Bob's receiver station. This restriction is modeled by lossy channels between Alice/Bob and Eve, where the transmissivity of such channels can, in principle, be bounded by monitoring techniques. An artefact of such lossy channels is the possibility of having {\it bypass} channels, those which are not accessible to Eve, but may not necessarily be characterized by the users either. This creates interesting, unexplored, scenarios for analyzing QKD security. In this paper, we obtain generic bounds on the key rate in the presence of bypass channels and apply them to continuous-variable QKD protocols with Gaussian encoding with direct and reverse reconciliation. We find regimes of operation in which the above restrictions on Eve can considerably improve system performance. We also develop customised bounds for several protocols in the BB84 family and show that, in certain regimes, even the simple protocol of BB84 with weak coherent pulses is able to offer positive key rates at high channel losses, which would otherwise be impossible under an unrestricted Eve. In this case the limitation on Eve would allow Alice to send signals with larger intensities than the optimal value under an ideal Eve, which effectively reduces the effective channel loss. In all these cases, the part of the transmitted signal that does not reach Eve can play a non-trivial role in specifying the achievable key rate. Our work opens up new security frameworks for spaceborne quantum communications systems. 
\end{abstract}

\maketitle

\section{Introduction}

Satellite-based quantum communications links \cite{Bonato_2009,MoliSanchez2009,MeyerScott_PRA1011,Bourgoin_2014,Boone_PRA2015,Hosseinidehaj_SatQKD2017,Bedington_NatComm2017, Nauerth_NatPhoton_2013,Wang_NatPhoton2013,Bourgoin_PRA2015,Vallone_PRL2015,Gunthner:17} can be part of a global solution to quantum key distribution (QKD) networks or, more generally, the quantum Internet \cite{Kimble:QuInternet,Pirandola:QInternet,Pirandola:AQCrypt, liorni2020quantum}. QKD provides two parties with a secret key that can be used in cryptographic protocols, such as one-time pad encryption. In the absence of practical quantum repeaters, however, point-to-point fiber-based QKD links are often limited to a distance of several hundred kilometres \cite{QKD833, Pittaluga2021, PRL509km, Zhang_Optica2018, 1000kmQKD}. In contrast, free-space QKD relying on ground-to-satellite, satellite-to-ground, and/or satellite-to-satellite quantum communications links can potentially offer secure key exchange over thousands of kilometers \cite{Liao_Nat_2017,Liao:ChinaAustria_PRL2018}. The successful launch of the Chinese QKD satellite in 2017, and the experiments carried out since then \cite{Liao_NatPhoton_2017,Liao_Nat_2017,Ren_Nat_2017, Liao:ChinaAustria_PRL2018}, has particularly been a game changer in bringing the field into a new exciting development phase while a substantial global effort is directed at finding practical solutions to the wide-scale deployment of QKD systems. That said, satellite-based quantum communications comes at an additional price for launching and operating possibly dedicated satellites, as well as with some restrictions on accessibility and the achievable key rate. This manuscript seeks solutions that can enhance the benefits reaped from investing in this technology by looking into relevant threat models {to a line-of-sight link, as in satellite-based QKD, while maintaining the key security features of QKD systems}.

To make the above vision possible, and, particularly, to deploy satellite-based QKD in large scales, certain technological challenges must be addressed. For instance, a secure satellite-based QKD system must combat loss and noise effects in the link. A satellite-to-ground link would also face additional challenges due to pointing errors and atmospheric turbulence, which impact system performance. Ultimate limits, as well as achievable rates of specific QKD protocols, have recently been investigated considering diffraction, extinction, background noise and fading in such links  \cite{Pirandola:FS2021,Pirandola:Sat2021,Ghalaii:StrongTB2021,Ghalaii:FSMDI2022}.  
Such analyses as well as recent experimental demonstrations suggest that a typical low-earth-orbit (LEO) satellite-to-ground link could suffer around 30-40 dB of loss for a modest-size receiver telescope \cite{Liao_Nat_2017}, and possibly with night operation only in order to minimize the background noise. This would imply that, under nominal security assumptions that give Eve maximum possible control over the channel, many QKD protocols may struggle to offer sufficiently high, if any, positive key rates. 

The above limitations are partly because of the assumptions made in our security analysis, e.g., that the channel in its entirety is assumed to be under the control of a potential eavesdropper. Whether such an assumption is necessary/realistic in satellite-based QKD, which relies on line-of-sight links, needs to be scrutinized. Relaxing this assumption could open up new opportunities that have been discounted, but which, if proved to be viable, could offer additional options for implementation and commercial exploitation. 

With the above idea in mind, recently, several works have addressed the security of satellite-based QKD in wiretap channels \cite{Ling-SatQKD-Man, Saikat_RestricedEve_PRApplied, Hugo_RestricedEve_PRApplied}, while earlier the security of QKD in the framework of physical layer security was considered \cite{Sasaki_2017}. The work in \cite{Ling-SatQKD-Man} considers a passive eavesdropping scenario for a wiretap channel \cite{Wyner-wiretap} and compares the key rate achievable under an unrestricted Eve for several QKD protocols with alternative schemes that they refer to as photon key distribution (PKD). They overall observe more resilience to noise in high-loss regimes for their PKD schemes, which allows them to cover longer distances. The work in \cite{Saikat_RestricedEve_PRApplied, Hugo_RestricedEve_PRApplied} considers the in-principle achievable key rate, in a wiretap channel, when only one of Alice and Bob measures their signal, and the other one holds onto a quantum state, on which they can in principle do an optimal measurement to maximise the key rate. They will then observe a boost in the key rate so long as the channel between Alice and Eve is lossier than that of Alice and Bob. In \cite{Hugo_RestricedEve_PRApplied} they further claim that by considering a protected zone around Alice (the satellite) and Bob (the ground station) they can ensure that the above condition holds if the presence of an eavesdropper in orbit can be ruled out. For the latter, they will then consider some constraints on celestial mechanics to show how difficult it would be for Eve to eavesdrop in this line-of-sight link.

In this manuscript, we study the security of prepare-and-measure (P\&M) satellite-based QKD for a restricted Eve without restricting ourselves to the case of the wiretap channel. This allows us to consider more generic cases and takes an important step toward having a verifiable set of assumptions. In the case of wiretap channels considered in \cite{Saikat_RestricedEve_PRApplied, Hugo_RestricedEve_PRApplied}, it will be difficult to ensure through experimental observations that the channel is indeed a wiretap channel{, or to specify the relevant channel parameters. One can potentially use monitoring techniques to rule out the possibility of having eavesdropping objects in the line-of-sight link. Even if we trust our employed monitoring technique, any such technique would, however, be bound by a certain resolution, and it is still possible that they miss objects of smaller than a certain size}. The potential users should then choose whether they are satisfied with these assumptions, or whether for provable security they wish to use a full QKD protocol.

{Note that the physical size of the devices an eavesdropper may have used has not been a matter of contention in conventional QKD systems. In conventional security proofs, we only care about the impact Eve may have on the quantum signals that Alice and Bob exchange, and they bound the leaked information to Eve based on the observations that they make in the quantum communication part of the protocol. By introducing monitoring techniques, we are not directly measuring the quantum interactions that Eve may have with the exchanged quantum signals, but instead we are trying to bound some classical aspects, such as size, of Eve's apparatus. While a super-powerful Eve could, in principle, fool our monitoring system too, in practice, this would add an additional layer of complexity to Eve's attack.}  

In our case, the primary assumption that we make about the potential eavesdropper is on the collection efficiency of her apparatus when it comes to interacting with the transmitted signal from Alice's telescope. This collection efficiency can then be bounded based on the size of devices that Eve has employed within the line-of-sight link. The corresponding size can, in principle, be bounded using reliable monitoring techniques that can be employed in parallel to quantum signal transmission. The same argument and methodology can be used to bound the loss between Eve and Bob. 

It is interesting to note that specifying the minimum loss that Alice's signal would go through before being collected by Eve does not specify the entire channel between Alice and Bob, and it is still possible that part of Alice's signal reaches Bob without going through Eve. This latter channel, which we refer to as a {\it bypass} channel, has a non-trivial role in the achievable key rate, and one of our key contributions here is to analyse QKD security in the presence of such bypass channels. Moreover, unlike the wiretap channel model, we can now consider scenarios where Alice-Eve loss is lower than that of Alice-Bob. By performing the security analysis under the above conditions, we can then bound the achievable key rate for a restricted Eve using a set of assumptions that are in-principle verifiable.  This turns out to offer better performance, as compared to unrestricted eavesdropping, without necessarily compromising on our security assumptions.

{Note that there is a difference between ``bypass'' channels, to which eavesdroppers do not have access although they may still indirectly use to their advantage, and ``side'' channels, which are assumed to be fully accessible to the eavesdropper. While the issue of side channels has been considered for years in QKD literature \cite{MDI-QKD1, side-channel1, side-channel2}, the topic of bypass channels is quite new, and we believe that this paper offers an intriguing formulation of this problem,  and then derives relevant generic and customized security bounds for the emerging settings.}

The key contributions of this paper are as follows: 
\begin{itemize}
\item We develop models for restricted eavesdropping whose elements can, in principle, be characterized using monitoring techniques;
\item We obtain generic bounds on achievable key rates in P\&M QKD setups in the presence of an uncharacterized bypass channel not accessible to Eve;
\item We show that, in certain practical regimes, such bounds enable continuous-variable (CV) QKD to offer positive key rates in satellite-based implementations; and
\item We develop customised bounds for discrete-variable (DV) QKD systems that rely on photon-number channels, and improve their performance under restricted eavesdropping.
\end{itemize}

The rest of this paper is organized as follows. In Sec.~\ref{sec:scenarios}, we describe our setting and the motivations behind the model we have adopted for the restricted Eve. In Sec.~\ref{Sec:SecProof}, we offer some generic results applicable to QKD protocols in the presence of bypass channels. We apply these results to CV-QKD protocols, in Sec.~\ref{cv-qkd}, and customize them to the case of DV QKD protocols, such as BB84 \cite{bennett1984brassard}, in Sec.~\ref{dv-qkd}. We conclude the paper in Sec.~\ref{SecCon} with some discussions on the relevance of the results obtained and the way forward for other cases not considered in this paper. 

\section{Generic Models for Restricted Eavesdropping}
\label{sec:scenarios}

In this section, we model the key restriction we consider in this work on potential eavesdroppers in a satellite-based QKD system. One of the distinctive features of a satellite link, as compared to a fiber link, is that it is a line-of-sight link. While it may not be possible, for a link of around 500~km of length in the LEO case, to fully monitor the channel between Alice and Bob, one can employ monitoring techniques, such as light detection and ranging (LIDAR), to detect objects of a certain minimum size along the path. In fact, the same system and the corresponding optics that are being used for tracking and acquisition purposes can also be used to detect unwanted objects along the beam. In free-space LIDAR, the power received by the detection site is proportional to the effective area of the object, and scales inversely with the power four of the distance between the object and the LIDAR source. If the collected power is below a certain noise threshold, we cannot conclusively declare detecting an object, but we might be able, at any given distance, to set a bound on the maximum size that any undetected object may have. In fact, our preliminary calculations suggest that for a 500-km-long satellite link, and for low-power LIDAR systems used at both Alice and Bob stations, with some nominal assumptions, the largest undetected object within the beam width of our LIDAR sources is around a few centimetres in diameter; see Appendix~\ref{Carlo}. This is important because, for any effective eavesdropping activity in the P\&M scenario, Eve requires (i) to somehow collect the signals transmitted by Alice, or reflect it to some other collection point, and/or (ii) to somehow be able to send her own signals towards Bob's receiver. In the satellite scenario, full power collection/reflection requires telescopes/optical tools of a certain size, corresponding to the beam width, and manipulation of Bob's receiver might need powerful laser sources, especially if Eve's source is not fully aligned with Bob's telescope. This implies that the combination of limited size telescopes/devices used in the line-of-sight link for Eve and a monitored/protected zone around Alice box could restrict Eve to only receiving a fraction of what Alice has sent. This would be the first departure point from a maximally powerful Eve. In the second case, where Eve cannot replace the channel between herself and Bob with an ideal channel, any active attack by Eve will be affected by potentially a lossy channel that the protection zone around the receiver would enforce. This could further restrict Eve in implementing her attack scenario.

\begin{figure}[t]
\includegraphics[width=0.5\textwidth-15pt]{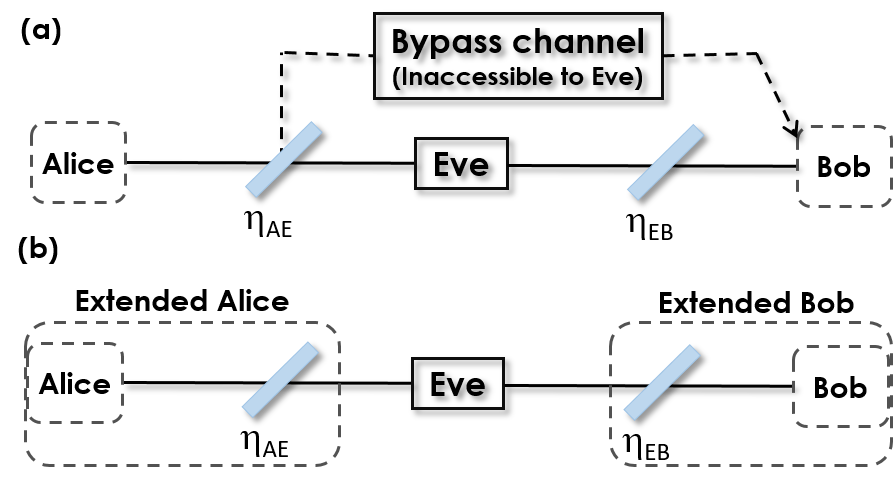}
\caption{(a) Restrictions imposed on Eve in terms of her collection efficiency, modelled by a beam splitter with transmissivity $\eta_{\rm AE}$, and her access to Bob's telescope via a beam splitter with transmissivity $\eta_{\rm EB}$. The part of the transmitted signal that does not go through Eve may still reach Bob via a bypass channel inaccessible to Eve. The signal lost at the second beam splitter, with transmissivity $\eta_{\rm EB}$, is assumed to be inaccessible to all parties. (b) A simplified model where the bypass channel in (a) is assumed to be not accessible to Bob. This assumption would effectively reduce the channel model in (a) to a typical prepare-and-measure QKD scenario with extended Alice's and Bob's boxes that contain some trusted lossy elements.
} 
\label{fig:chmodel}
\end{figure} 

In this work, we model the restrictions explained above, which can, in principle, be characterized by the employed monitoring systems, by lossy channels between Alice and Eve, and between Eve and Bob. In particular, as shown in Fig.~\ref{fig:chmodel}(a), we assume that a lossy channel with transmissivity $\eta_{\rm AE}$ connects Alice to Eve, and Eve has no access to the signals lost in this channel. Note that part of the lost signal can still reach Bob, and we cannot discount this possibility. This creates an interesting QKD scenario, where, in addition to the channel controlled by Eve, there is a {\it bypass} channel via which some signals can reach Bob. Eve has no access to this bypass channel, but Alice and Bob cannot necessarily characterize this channel either. The study of QKD security in the presence of such a bypass channel would generate interesting scenarios that we analyse in this paper. Similarly, we assume that every signal sent by Eve to Bob would go through a lossy channel with transmissivity  $\eta_{\rm EB}$, where Eve (and Bob) has no access to the lost signals on this channel. We do not impose any other restrictions on Eve except being bound by the laws of quantum mechanics. We investigate how these two restrictions affect the performance of a QKD system ran on such a link.

\begin{figure*}
\includegraphics[width=15cm]{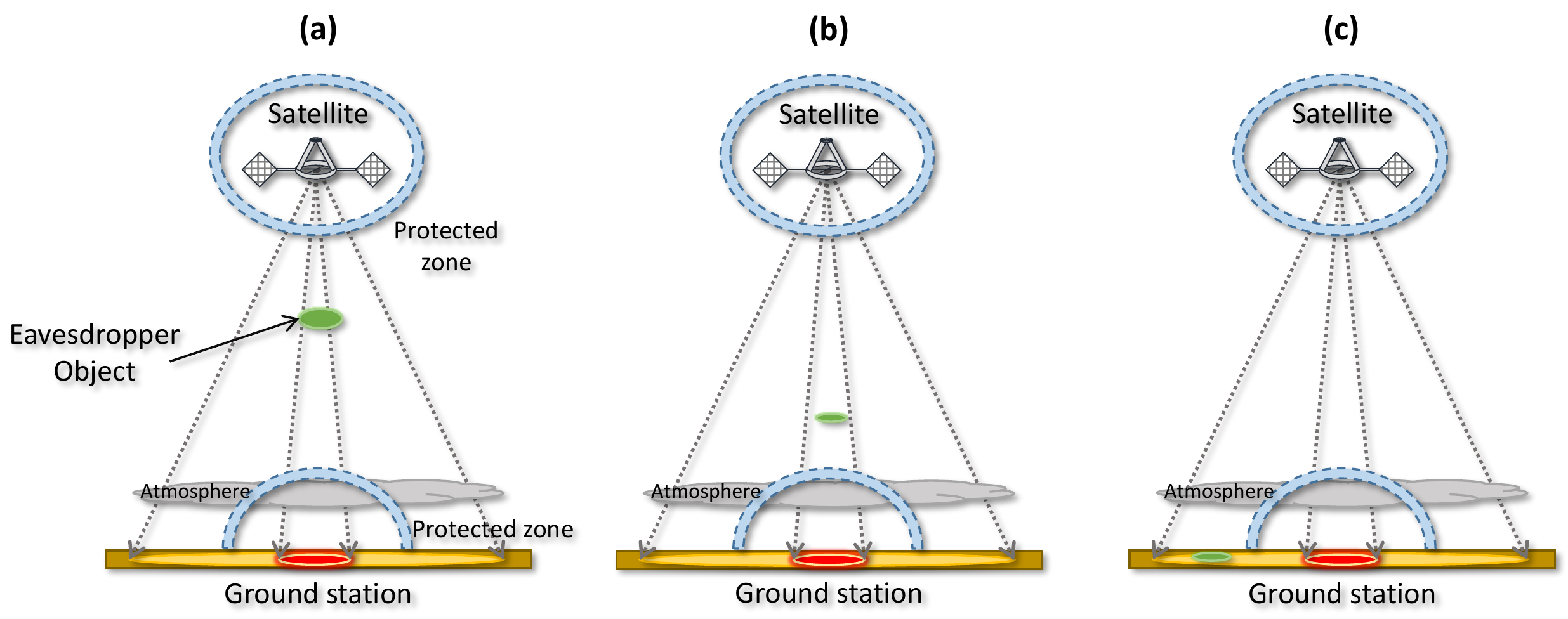}
\caption{Schematic view of a satellite-to-ground QKD link, with different restricted eavesdropping scenarios: (a) a semi-powerful Eve who, while does not capture the entire beam sent by the satellite (Alice), has access to the part that will be collected by the ground station (Bob); (b) An Eve with a telescope too small to capture the entire signal that reaches Bob; and (c) A passive Eve in a wiretap channel.} 
\label{fig:setup}
\end{figure*}

There are different scenarios that one can consider with the above generic restrictions. One possible scenario, shown in Fig.~\ref{fig:setup}(a), is when Eve's telescope is sufficiently large to capture all signals that would end up on Bob's telescope, but not necessarily large enough to capture the entire signal sent by Alice. This case corresponds to $\eta_{\rm AE}<1$, but possibly with $\eta_{\rm EB}$ close to one. Please note that when we are speaking of Eve, she is not restricted to operate only from one point in space. Another possibility is when Eve's telescope is assumed to be too small to capture the entire signal that would be received by Bob, in which case part of Alice's signal may reach Bob without Eve's intervention; see Fig.~\ref{fig:setup}(b). This case would result in intriguing scenarios especially when $\eta_{\rm AE} \ll 1$. We look at how we can capitalise on this restriction to increase the secret key rate in forthcoming sections. One last case, shown in Fig.~\ref{fig:setup}(c), is for when Eve is simply a passive receiver of Alice's signal without sending anything to Bob. This case corresponds to a small $\eta_{\rm AE}$ and $\eta_{\rm EB}=0$, and captures a passive attack on a wiretap channel \cite{Ling-SatQKD-Man}. These are just a few examples, but the important point is that the generic model proposed here for a natural restriction on Eve can capture many practical cases that could happen in reality, as well as the few cases considered thus far in the literature \cite{Ling-SatQKD-Man, Saikat_RestricedEve_PRApplied, Hugo_RestricedEve_PRApplied}.

Our objective in this paper is to find bounds on the secret key generation rate under the assumption that $\eta_{\rm AE}$ and $\eta_{\rm EB}$ are known to Alice and Bob. We separate the issue of how, in practice, we can find an upper bound for these parameters from the security proof that follows once this restrictive assumption is used. The latter will be discussed in Sec.~\ref{Sec:SecProof}, with particular examples on CV and DV QKD in Secs.~\ref{cv-qkd} and \ref{dv-qkd}, respectively. For the former, in Appendix~\ref{Carlo}, we consider a simple model to calculate the reflected power from an object (or a collection of objects with a similar effective size) with a certain reflectivity, in the line-of-sight link, assuming that a LIDAR system has been employed on both the satellite and ground station. If our LIDAR system detects an object of a certain size, we can then use that to bound $\eta_{\rm AE}$ and $\eta_{\rm EB}$. Even if the LIDAR systems do not detect any object, by making some nominal assumptions on the power budget on satellite and earth, the sensitivity of the LIDAR system, and the reflectivity of space objects, we can then find the maximum object size that may remain undetected by our LIDAR systems, and then accordingly upper bound $\eta_{\rm AE}$ and $\eta_{\rm EB}$. This preliminary analysis suggests that, in nominal working conditions, $\eta_{\rm EB}$ is greater than $\eta_{\rm AE}$, and can be close to 1, whereas $\eta_{\rm AE}$ can remain small. In our analysis in Secs.~\ref{cv-qkd} and \ref{dv-qkd} we then only consider the special case of $\eta_{\rm AE}<1$ at $\eta_{\rm EB}=1$, which is of practical interest.

In what follows, we first find some generic results for the key rate of the setup in Fig.~\ref{fig:chmodel}(a). Throughout the paper, the satellite is assumed to have the QKD encoder and the ground station would decode the received signals. We therefore mainly focus on prepare-and-measure schemes in the forthcoming sections. In particular, we consider the BB84 protocol with different types of sources, and CV-QKD with Gaussian encoding. One interesting point about the restricted Eve scenario is the possibility of designing new protocols that capitalize on Eve's imposed restrictions. For instance, as shown in \cite{Ling-SatQKD-Man}, in the case of a passive Eve, one can relax the requirement for using two mutually unbiased bases to come up with simpler protocols. Or, in the case of an ideal single-photon source (SPS) with a passive Eve, no privacy amplification may be needed \cite{Huttner-SatQKD}. In our setting, the bypass channel in Fig.~\ref{fig:chmodel}(a) can play a non-trivial role in determining the key rate, as we investigate next.

\section{Security Proof}
\label{Sec:SecProof}
In this section we aim at finding generic bounds on the secret key generation rate for the setup in Fig.~\ref{fig:chmodel}(a). { The key assumption in our analysis is that Alice and Bob can reliably characterize parameters $\eta_{\rm AE}$ and $\eta_{\rm EB}$ in Fig.~\ref{fig:chmodel}(a). Otherwise, we do not need to know the nature of the bypass channel, and the bypass channel, while inaccessible to Eve, remains uncharacterized by Alice and Bob. This is in contrast with what typically assumed in physical layer security, or earlier work on restricted eavesdropping, in which certain channel models are assumed \cite{Ling-SatQKD-Man, Saikat_RestricedEve_PRApplied, Hugo_RestricedEve_PRApplied, Sasaki_2017}.} 

To get some insight into the setting of Fig.~\ref{fig:chmodel}(a), one simplifying assumption, as shown in Fig.~\ref{fig:chmodel}(b), is to ignore the bypass channel and assume that no information would reach Bob via the bypass channel. 
{This assumption would effectively reduce the channel model in Fig.~\ref{fig:chmodel}(a) to a typical prepare-and-measure QKD scenario with extended Alice’s and Bob’s boxes that contain some trusted lossy elements. The secret key rate calculations in Fig.~\ref{fig:chmodel}(b) would then reduce to modifying existing security proofs to account for the trusted loss in the channel. This would provide us with a reference point to which we can compare the key rate of QKD systems with bypass channels as in Fig.~\ref{fig:chmodel}(a). On the one hand, having a bypass channel that Eve has no access to may suggest that Alice and Bob can share their secret key more easily implying that the key rate in scenario (b) is a lower bound to that of (a). On the other hand, because the bypass channel is not fully characterized by Alice and Bob, they need to consider the worst-case scenario, compatible with their observations, in which case Eve may end up being the beneficiary of the bypass channel.

One of our key contributions is to prove that, under a given set of experimental observations, the key rate of Fig.~\ref{fig:chmodel}(a) is always upper bounded by that of Fig.~\ref{fig:chmodel}(b). We label this result as Theorem 1 and will prove it in this section. That said, by properly formulating the problem, we can also see how the other intuition comes into play, and, under what scenarios, it may prevail. Lemma 1 will capture this other result. But, first, let us diligently formulate the two settings in Fig.~\ref{fig:chmodel}.}

In Figs.~\ref{fig:scenarios}(a) and (b), we have presented generic attack models, in the entanglement-based picture, for, respectively, the scenarios in Figs.~\ref{fig:chmodel}(a) and (b). Here, $|\psi_{AB}\rangle$ represents the initial bipartite entangled state generated by Alice, where one of its components is measured by measurement operator $M_A$ to give the classical outcome $X$, and its other component is sent to Bob. In Fig.~\ref{fig:scenarios}, we have used the same notation for the field modes at the input and output of a quantum operation. For instance, mode $B$ would go through the initial beam splitter, and then through Eve's system, followed by the second beam splitter before entering Bob's telescope, modelled by operator $\mathcal{E}_T$, and measurement operator $M_B$ resulting in a classical variable $Y$. The measurement operator $M_B$ effectively models the corresponding QKD measurements in the respective QKD protocol. Given that the bypass channel and Eve-controlled channels represent two independent spatial modes, the operator $\mathcal{E}_T$ effectively combines these two modes to generate outcome $Y$. For a physical telescope, these two modes are defined by what the telescope actually collects. In that case, this operation has to model a unitary evolution. We therefore assume $\mathcal{E}_T$ is a unitary map, in which case we need to introduce a second output mode, which we have denoted by $F_0$. In our setup, mode $F_0$ is not accessible to Bob, but it would be interesting to see what, in principle, is achievable for Alice and Bob if $F_0$ is available to Bob. Lemma 1 below considers this case. Other important components of Figs.~\ref{fig:scenarios}(a) and (b) are complete positive and trace preserving (CPTP) maps $\mathcal{E}$ and $\mathcal{E}'$, which, respectively, model the channel controlled by Eve and the bypass channel, with pure input states denoted by $|\psi_E\rangle$ and $|\psi_F\rangle$. In order to match the model in Fig.~\ref{fig:scenarios}(b) with that of Fig.~\ref{fig:chmodel}(b), we have introduced a trivial map $\mathcal{E}_V$ that maps every incoming state to the vacuum state $\ket{0}$. More specifically, the map $\mathcal{E}_V$ is a CPTP map with the following Kraus representation:
\begin{align}
\mathcal{E}_V(\rho) = \sum_{i} K_i \rho K_i^\dag \quad,\quad K_i:=\ketbra{0}{e_i} \label{Kraus-rep}
\end{align}
with $\{\ket{e_i}\}$ being an orthonormal basis for the Hilbert space where the input state $\rho$ lies in. This operation ensures that nothing but the vacuum state would be transferred via the bypass channel, which corresponds to the simplified scenario in Fig.~\ref{fig:chmodel}(b). Finally, the second input to both beam splitters in Fig.~\ref{fig:scenarios} is the vacuum state to model a lossy channel.

\begin{figure}[t]
\includegraphics[width=0.5\textwidth-15pt]{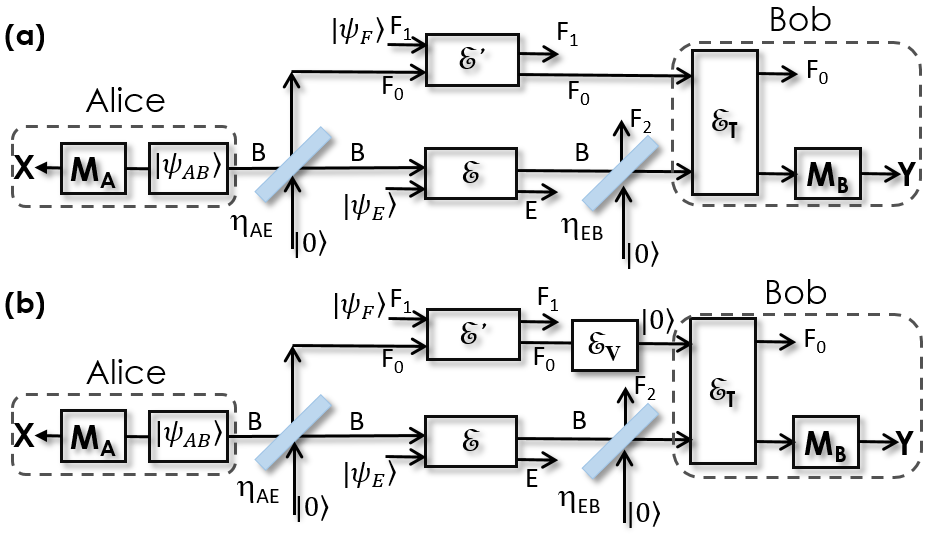}
\caption{(a) A generic attack model for a QKD system under restrictive assumptions on the collection efficiency and transmission efficiency of Eve's apparatus; (b) The attack model assuming that the bypass channel includes an infinitely high loss and only produces the vacuum state at its output. Notations are defined in the text.}
\label{fig:scenarios}
\end{figure}

For the above detailed settings, we now investigate how the key rate achievable in Fig.~\ref{fig:scenarios}(a), which corresponds to the main restrictions imposed on Eve in our work, compares with that of Fig.~\ref{fig:scenarios}(b), which further simplifies the channel and makes additional assumptions. {As discussed earlier, because Eve has no access to the bypasss channel, one may expect that the former cannot be lower than the latter. In Lemma~1, we prove that this intuition is correct in the case of direct reconciliation, provided that Eve's attack (map $\mathcal{E}$) is fixed in both scenarios of Fig.~\ref{fig:scenarios}(a) and Fig.~\ref{fig:scenarios}(b) and mode $F_0$ is available to Bob. However, from a security perspective, we cannot ensure that Eve would perform the same attack independently of the physical channel(s) linking Alice and Bob. Interestingly, when allowing for the worst-case attack by Eve in each scenario of Fig.~\ref{fig:scenarios}, and conditioned on the observed parameters in the QKD experiment, the achievable key rate in Fig.~\ref{fig:scenarios}(a) turns out to be upper bounded by that of Fig.~\ref{fig:scenarios}(b), as we prove in Theorem~1.} Note that the bypass channel $\mathcal{E}'$ is not necessarily known to Alice and Bob.

Let us first consider the case where mode $F_0$ is available to Bob {and Eve's attack is identical in both scenarios of Fig.~\ref{fig:scenarios}}.

\begin{Lmm} \label{lmm:invertibleMprime}
For a quantum Bob with access to modes $B$ and $F_0$, and a unitary map $\mathcal{E}_T$, the in-principle achievable asymptotic key rates $r_a$ and $r_b$, with one-way {\em direct reconciliation}, corresponding, respectively, to the setups in Figs.~\ref{fig:scenarios}(a) and (b), satisfy
\begin{align}
r_b \leq r_a \label{to-prove}.
\end{align}
\end{Lmm}

The proof is given in Appendix~\ref{app:lemma1proof}. {The proof of Lemma~1 hinges on the fact that the scenario in Fig.~\ref{fig:scenarios}(b) can be recovered from Fig.~\ref{fig:scenarios}(a) by applying an additional map on Bob's systems---effectively, the extra map $\mathcal{E}_V$ that maps everything to the vacuum. Such a map does not affect Eve's uncertainty about Alice's $X$ outcomes while it increases Bob's uncertainty, by possibly increasing the quantum bit error rate (QBER) in a QKD experiment. This implies that, under conditions of Lemma~1, the in-principle achievable key rate in Fig.~\ref{fig:scenarios}(b) should not be higher than that of Fig.~\ref{fig:scenarios}(a).}

{The result of Lemma~1, however, holds for a quantum Bob under fixed attack by Eve performed in the two scenarios of Fig.~\ref{fig:scenarios} and might not be of use when evaluating the secret key rate produced in a given QKD experiment. As a matter of fact, in a QKD experiment, what we are interested in is a bound on the leaked information to Eve conditioned on the set of observations made in the corresponding QKD experiment, in either configurations in Fig.~\ref{fig:scenarios}. Considering that scenario (b) is equal to scenario (a) except for possibly an additional noise-increasing map, by fixing the observed amount of noise, we may conclude that the required attack by Eve can be less powerful in (b) than in (a), such that the resulting noise is effectively the same in the two scenarios. A less powerful attack could amount to less information leaked to Eve, hence a higher secret key rate in the case of Fig.~\ref{fig:scenarios} (b). This leads us to the opposite conclusion from what we draw in Lemma~1, namely, that, in the P\&M QKD setting, the secret key rate in Fig.~\ref{fig:scenarios} (a) cannot be larger than that of  Fig.~\ref{fig:scenarios} (b). An alternative way to look at this problem is that, from Alice and Bob's point of view, they have to find the worst case attack in the space spanned by valid choices of $\{\mathcal{E}, \mathcal{E}'\}$, for Fig.~\ref{fig:scenarios} (a), and in the space of $\{\mathcal{E}, \mathcal{E}_V\}$ for Fig.~\ref{fig:scenarios} (b). The latter turns out to be a subset of the former, which implies that Eve might come up with a more effective attack in the setup of  Fig.~\ref{fig:scenarios} (a).  We formalize this argument in the following theorem, which rigorously proves the above insight in the finite-key scenario and for both direct and reverse reconciliation cases.}

\begin{Thm}
\label{Thm:DR-QKD}
Consider an $\varepsilon$-secure QKD protocol, with one-way direct (or reverse) information reconciliation and $\varepsilon=2\bar{\varepsilon}+ \varepsilon_{\rm EC}+\varepsilon_{\rm PA}$, where $\varepsilon_{\rm EC}$ and $\varepsilon_{\rm PA}$ are, respectively, the security parameters for the error correction and privacy amplification steps. Let $n$ be the number of signals used for key generation and $\{Q^{\rm obs}_1,Q^{\rm obs}_2,\dots\}$ be the observed parameters by Alice and Bob in the parameter-estimation rounds. Then, the achievable secret key rates $R_a$ and $R_b$ of scenarios (a) and (b) in Fig.~\ref{fig:scenarios}, respectively, obtained with the above protocol in the finite-key regime satisfy:
\begin{align}
R_a \leq R_b \label{theorem1-result}.
\end{align}
\end{Thm}

\begin{proof}
{The claim directly follows from the definitions of achievable secret key rate for scenario (a) and (b) in the finite-key regime. To see this, let us first consider the state $\rho_{X^n Y^n E}$ representing the raw keys of Alice and Bob, together with Eve's quantum side information. For simplicity, we assume that Bob assigns a random outcome in the case of no detection in a key generation round. A similar proof would hold in the case where Alice and Bob apply a sifting map to their outcomes in order to discard the rounds where Bob had no detection. Let us denote the initial state of all subsystems, before any map is applied, by $\rho$ given by:
\begin{align}
\rho:= & \ketbra{\psi_{AB}}{\psi_{AB}}^{\otimes n}\otimes \ketbra{0}{0}_{F_0}^{\otimes n} \otimes \ketbra{\psi_F}{\psi_F} \\ \notag 
& \otimes \ketbra{\psi_E}{\psi_E} \otimes \ketbra{0}{0}^{\otimes n}_{F_2}.
\end{align}
Then, for scenario (a), we have
\begin{align}
\begin{split}
    \rho^{(\mathcal{E},\mathcal{E}')}_{X^n Y^n E}=\Tr_{F_0 F_1 F_2}  \Bigl[
    & M_B \circ \mathcal{E}_T\circ\mathcal{B}_{\eta_{\rm EB}}\circ\mathcal{E}'  \\
    & \circ \mathcal{E}\circ\mathcal{B}_{\eta_{\rm AE}}\circ M_A (\rho) \Bigr], 
\end{split}
\label{rhoXYEa}
\end{align}
and for scenario (b),
\begin{align}
\begin{split}
    \rho^{(\mathcal{E})}_{X^n Y^n E}=\Tr_{F_0 F_1 F_2} \Bigl[
    & M_B \circ \mathcal{E}_T\circ\mathcal{B}_{\eta_{\rm EB}}\circ\mathcal{E}_V \\
    & \circ \mathcal{E}\circ\mathcal{B}_{\eta_{\rm AE}}\circ M_A (\rho) \Bigr], 
\end{split}
\label{rhoXYEb}
\end{align}
where we denote the maps of the two beam splitters by $\mathcal{B}_{\eta_{\rm AE}}$ and $\mathcal{B}_{\eta_{\rm EB}}$ and discarded the map $\mathcal{E}'$ in \eqref{rhoXYEb} since it would have no effect on the state. Then, the state in \cref{rhoXYEb} can be obtained from \cref{rhoXYEa} by replacing $\mathcal{E}'$ with $\mathcal{E}_V$, that is, $\rho^{(\mathcal{E})}_{X^n Y^n E}=\rho^{(\mathcal{E},\mathcal{E}_V)}_{X^n Y^n E}$ .}

{For scenario (a), the achievable secret key rate obtained from the $n$ detected key-generation rounds, in the case of direct reconciliation, is given by \cite{ScaraniRennerPRL}:
\begin{align}
\begin{split}
R_a = \frac{1}{n} \biggl[
& \min_{(\mathcal{E},\mathcal{E}')\in\mathcal{S}(\{Q^{\rm obs}_1,Q^{\rm obs}_2,\dots\},\bar{\varepsilon})} H^{\bar{\varepsilon}}_{\mathrm{min}}   (X^n|E)_{\rho^{(\mathcal{E},\mathcal{E}')}} \\
& - I_\mathrm{EC} - \log_2 \frac{2}{\varepsilon_{\rm EC}} -2\log_2 \frac{1}{2\varepsilon_{\rm PA}} \biggr], 
\end{split}
\label{finitekeya}
\end{align}
where the minimization is performed over all possible attacks by Eve, $\mathcal{E}$, and all possible actions of the bypass channel,  $\mathcal{E}'$, compatible with the observed parameters, while $I_\mathrm{EC}$ is the amount of error-correction information publicly revealed by Alice and $H^{\bar{\varepsilon}}_{\mathrm{min}}$ is $\bar{\varepsilon}$-smooth min entropy function. More specifically, the set $\mathcal{S}(\{Q^{\rm obs}_1,Q^{\rm obs}_2,\dots\},\bar{\varepsilon})$ contains all pairs of maps $(\mathcal{E},\mathcal{E}')$ such that the parameters $\{Q^n_1,Q^n_2,\dots\}$, computed from the resulting state $\rho^{(\mathcal{E},\mathcal{E}')}_{X^n Y^n}$ in \cref{rhoXYEa}, are close to the observed parameter values $\{Q^{\rm obs}_1,Q^{\rm obs}_2,\dots\}$, except for a small probability fixed by $\bar{\varepsilon}$.}

{Similarly, for scenario (b), the achievable secret key rate is given by:
\begin{align}
\begin{split}
R_b = \frac{1}{n} \biggl[
& \min_{\mathcal{E}\in\mathcal{T}(\{Q^{\rm obs}_1,Q^{\rm obs}_2,\dots\},\bar{\varepsilon})} H^{\bar{\varepsilon}}_{\mathrm{min}}(X^n|E)_{\rho^{(\mathcal{E})}} \\
& - I_\mathrm{EC} - \log_2 \frac{2}{\varepsilon_{\rm EC}} -2\log_2 \frac{1}{2\varepsilon_{\rm PA}} \biggr], 
\end{split}
\label{finitekeyb}
\end{align}
where in this case the set $\mathcal{T}(\{Q^{\rm obs}_1,Q^{\rm obs}_2,\dots\},\bar{\varepsilon})$ contains all possible maps $\mathcal{E}$ such that the parameters $\{Q^n_1,Q^n_2,\dots\}$, computed from the resulting state $\rho^{(\mathcal{E})}_{X^n Y^n}$ in \eqref{rhoXYEb}, are close to the observed values $\{Q^{\rm obs}_1,Q^{\rm obs}_2,\dots\}$, except for a small probability fixed by $\bar{\varepsilon}$.}

For a fixed set of values $\{Q^{\rm obs}_1,Q^{\rm obs}_2,\dots,\bar{\varepsilon}\}$, Eqs.~\eqref{finitekeya} and \eqref{finitekeyb} are identical expect for their smooth min-entropy terms. {Moreover, we observe that the minimization set in \eqref{finitekeyb} is a subset of the minimization set in \eqref{finitekeya}. In particular, the smooth min entropy term in \eqref{finitekeyb} is calculated for $\rho^{(\mathcal{E})}_{X^n Y^n}=\rho^{(\mathcal{E},\mathcal{E}_V)}_{X^n Y^n}$, which is a subset of all the states $\rho^{(\mathcal{E},\mathcal{E}')}_{X^n Y^n}$ that are considered in \cref{finitekeya}. In other words, we have: $\mathcal{T}\times \{\mathcal{E}_V\}\subseteq \mathcal{S}$. We, therefore, conclude that the minimization in \cref{finitekeya} can only produce a smaller or equal rate than the minimization in \cref{finitekeyb}, thus proving the claim that $R_a \leq R_b$.}

{Note that the same proof can straightforwardly be extended to the reverse reconciliation case, by replacing Alice's raw key $X^n$ with Bob's raw key $Y^n$ in the smooth min-entropy terms. We again observe that, by minimizing the achievable key rate over the uncharacterized maps of the setups in Fig.~\ref{fig:scenarios}, namely, $\mathcal{E}$ and $\mathcal{E}'$ in (a) and $\mathcal{E}$ in (b), scenario (b) can be seen as a particular case of scenario (a). Thus, the optimal key rate in (a) should be smaller than or equal to the optimal key rate in (b). However, this also suggests that a partial characterization of the map $\mathcal{E}'$ in the bypass channel would prevent us from viewing (b) as a particular case of (a), leading to a potentially different relation between the key rates $R_a$ and $R_b$.}
\end{proof}

Theorem \ref{Thm:DR-QKD} provides an easy way to obtain upper bounds on the key rate in the generic setup of Fig.~\ref{fig:scenarios}(a), which includes a bypass channel, using existing techniques and bounds for the setup of Fig.~\ref{fig:scenarios}(b), which includes extended Alice and Bob boxes. While this is an important result, in QKD, we are often interested in lower bounds on the key rate, by which we can specify the required amount of privacy amplification in a real experiment. In the following sections, we will further study the relationship between such lower and upper bounds in the case of certain CV and DV-QKD protocols. In particular, we numerically check in the case of CV-QKD how the two bounds are close to, or deviate from, each other in certain practical scenarios. In the case of DV-QKD, we also use the photon-number nature of the channel in certain BB84 protocols to come up with customized lower bounds in the setups with a bypass channel.

{An alternative way to lower bound the min-entropy term in \cref{finitekeya}, in the direct reconciliation case, is to calculate $H^{\bar\varepsilon}_{\mathrm{min}}(X^n|B')$, where, in Figs.~\ref{fig:scenarios}(a) and (b), $B'$ represents mode $B$ right after the first beam splitter, which is in the state given by $\rho_{B'}= \Tr_{AF_0}[\mathcal{B}_{\eta_{\rm AE}}(\ketbra{\psi_{AB}}{\psi_{AB}}^{\otimes n} \otimes \ketbra{0}{0}_{F_0}^{\otimes n} )]$. To prove this, consider that the min-entropy in \cref{finitekeya} is computed on the state in \cref{rhoXYEa} where the system $Y^n$ is traced out. This allows us to simplify some of the quantum maps in the state in \cref{rhoXYEa} since they have no effect once the systems on which they act are traced out. We thus have that the min-entropy term in \cref{finitekeya} is computed on the following state:
\begin{align}
\begin{split}
\rho_{X^n E}= \Tr_{F_0  B} \Bigl[
& \mathcal{E}\circ\mathcal{B}_{\eta_{\rm AE}}\circ M_A (\ketbra{\psi_{AB}}{\psi_{AB}}^{\otimes n} \\
& \otimes \ketbra{0}{0}_{F_0}^{\otimes n} \otimes \ketbra{\psi_E}{\psi_E}) \Bigr], 
\end{split}
\label{rhoXEa}
\end{align}
Then, we can use the strong subadditivity of the smooth min-entropy function \cite{Tomamichel-book} to obtain the following lower bound:
\begin{align}
H^{\bar\varepsilon}_{\mathrm{min}}(X^n|E) &\geq H^{\bar\varepsilon}_{\mathrm{min}}(X^n|BE) , \label{ineq1}
\end{align}
where the entropy on the right hand side is computed on the state:
\begin{align}
\begin{split}
    \rho_{X^n B E}=\Tr_{F_0 }  \Bigl[
    & \mathcal{E}\circ\mathcal{B}_{\eta_{\rm AE}}\circ M_A (\ketbra{\psi_{AB}}{\psi_{AB}}^{\otimes n}  \\
    & \otimes \ketbra{0}{0}_{F_0}^{\otimes n} \otimes \ketbra{\psi_E}{\psi_E}) \Bigr]. 
\end{split}
\label{rhoXBEa}
\end{align}
By using the data-processing inequality \cite{Tomamichel-book}, the entropy can be further bounded as follows:
 \begin{align}
    H^{\bar\varepsilon}_{\mathrm{min}}(X^n|BE) \geq H^{\bar\varepsilon}_{\mathrm{min}}(X^n|B'E), \label{ineq2}
 \end{align}
where the entropy on the right hand side is now computed on the state without eavesdropper's map $\mathcal{E}$, i.e.,
\begin{align}
    \rho_{X^n B' E}= & \Tr_{F_0 } \Bigl[\mathcal{B}_{\eta_{\rm AE}}\circ M_A (\ketbra{\psi_{AB}}{\psi_{AB}}^{\otimes n}\otimes \ketbra{0}{0}_{F_0}^{\otimes n})\Bigr] \notag \\
    & \otimes \ketbra{\psi_E}{\psi_E}\label{rhoXBprimeEa}.
\end{align}

Because system $E$ is separate from all other systems in \cref{rhoXBprimeEa}, it follows that its contribution to the conditional entropy vanishes, i.e. $H^{\bar\varepsilon}_{\mathrm{min}}(X^n|B'E)=  H^{\bar\varepsilon}_{\mathrm{min}}(X^n|B')$. By combining this with \cref{ineq1} and \cref{ineq2}, we then obtain
\begin{align}
    H^{\bar\varepsilon}_{\mathrm{min}}(X^n|E) &\geq H^{\bar\varepsilon}_{\mathrm{min}}(X^n|B')
    \label{eq:obv-bound},
\end{align}
which proves our claim. Note that, in certain regimes of operation, \cref{eq:obv-bound} would allow us to obtain an effective lower bound on the key rate, as we will see in the CV-QKD section.}

\section{CV-QKD with restricted Eve} 
\label{cv-qkd}
Here, we focus on continuous-variable QKD protocols, in which data is encoded on the quadratures of light. We consider a particular protocol in the family of GG02 protocols \cite{Grosshans_GG02_PRL,Grosshans_GG02_Nature}, in which Alice uses Gaussian encoding, and Bob performs homodyne detection. CV-QKD is not an obvious choice when it comes to highly lossy channels \cite{Pirandola-MDI} such as the ones we may face in the satellite-based QKD scenario. But, for that very reason, it is a particularly interesting case to study because, in our setting, the initially trusted loss $\eta_{\rm AE}$ could alleviate some of the problems that CV-QKD faces in high-loss channels. Note that, by using the fading nature of the atmospheric part of the link \cite{Ruppert_2019}, {along with relevant binning or clustering techniques,} it might also be possible to find working regimes of operation for satellite-based CV-QKD \cite{SatCVQKD_npjQI,Derkach_2021, MalaneyQE}. 
In this work, however, we only focus on the benefits we may reap by imposing access restrictions on Eve, particularly, at the transmitter end, by assuming $\eta_{\rm AE} \leq 1$ while $\eta_{\rm EB} = 1$. For the same reason, we only focus on the asymptotic case, which also makes the analysis a bit easier to follow.

To be able to obtain concrete results, for the most of this section, we study a special case of the setup in \cref{fig:scenarios}(a), which we expect to encounter in practice. A schematic diagram of this case is given in \cref{fig:cv_telmodel}, in which the bypass channel is modelled as a pure loss channel with transmissivity $\eta_{\rm S}$. This is a reasonable assumption considering scenarios we may face in practice. Alternatively, a thermal-loss channel could have been assumed for the bypass channel, but as we will see later the insights we obtain into the effects of the bypass channel on the performance would not majorly change. The second assumption is in modelling the telescope action as a coupling beam splitter with transmissivity $\eta_{\rm T}$. As we will show in Appendix \ref{App:telescope}, this is partly the result of the mode definitions in \cref{fig:scenarios}(a), and partly because of the light collecting nature of a telescope. Finally, whenever Eve's action needs to be explicitly modelled, we assume Eve is implementing an entangling cloner attack. This would implicitly imply that the channel controlled by Eve is of thermal-loss nature. This may not be necessarily the case, especially in our setting where the bypass channel can offer other pathways to the receiver. But, again, it is what we may expect to be the case in a realistic scenario, and it also considerably reduces the search space when we look for worst-case configurations.

\begin{figure}[t]
\includegraphics[width=0.5\textwidth-15pt]{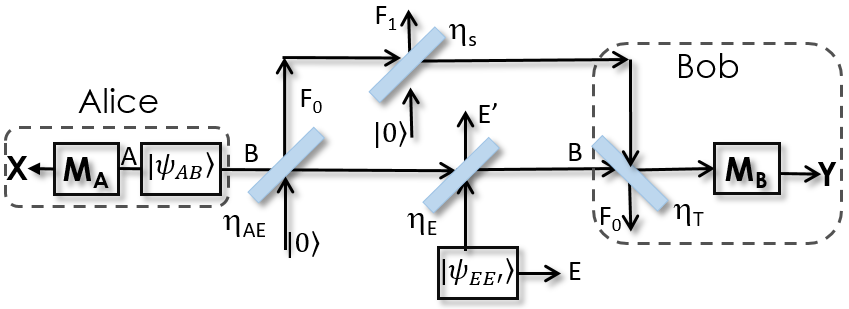}
\caption{A special setting for the setup of \cref{fig:scenarios}(a), where the bypass channel and the telescope actions are modeled by beam splitters. Eve's attack has also been modelled using an entangling cloner assuming that the Eve controlled section of the channel is also lossy.} 
\label{fig:cv_telmodel}
\end{figure}

The key rate of a CV-QKD protocol, in the asymptotic limit of infinitely many signals, in the direct reconciliation (DR) and reverse reconciliation (RR) cases, are, respectively, given by 
\begin{align}
\label{rate_DR}
K_{\rm DR} &=\beta I_{AB} - \chi_{AE},  \\
K_{\rm RR} &=\beta I_{AB} - \chi_{BE},
\label{rate_RR}
\end{align}
where $\beta$ is the reconciliation efficiency, $I_{AB}$ is the mutual information between Alice and Bob, and $\chi_{AE}$ ($\chi_{BE}$) is the Holevo information between Alice (Bob) and Eve. Under optimal collective Gaussian attacks \cite{Garcia-Patron_OptGaussAttacks,Navascues_OptGaussAttacks,Pirandola:ColGausAttacks_PRL2008}, the mutual information and Holevo information terms can both be bounded by using the covariance matrix (CM) of Alice, Bob, and Eve in the equivalent entanglement based picture of the protocol. In the unrestricted Eve scenario, it can be assumed that Eve holds a purification of Alice and Bob joint states. This enables us to calculate all relevant terms just as a function of the CM of Alice and Bob, which can directly be measured in the experiment. In the restricted Eve scenario, however, this purification assumption does not hold as there are other modes, such as $F_0$, $F_1$, and $F_2$ in \cref{fig:scenarios}, that are not accessible to any of the parties. This would require us to redo some of the calculations in the simulation cases we consider in this section.

Throughout this section, we assume that the measured CM by Alice and Bob implies a channel with a {total equivalent} excess noise, at the transmitter end, $\xi$, and a total transmissivity $T_{\rm eq} = \eta_{\rm ch} \eta_d$, where $\eta_{\rm d}$ is the receiver efficiency, corresponding to the measurement operator $M_B$, which is a trusted source of loss that can be characterized by the users{, and $\eta_{\rm ch}$, representing the channel transmissivity, is defined as the ratio between the two observed parameters $T_{\rm eq}$ and $\eta_d$}. {Note that in the asymptotic case considered in our analysis, the observed values for $T_{\rm eq}$ and $\xi$ effectively represent the corresponding average values for, respectively, transmissivity and excess noise, over the entire set of exchanged quantum states. This does not imply or require that the channel parameters need to be fixated throughout the experiment. In fact, in the satellite-to-ground channels, the turbulence effect can indeed result in a fading channel with a time-dependent gain. But, our security proof only relies on the average values derived from our observations, based on which the amount of information leaked to Eve can be bounded. Considering that, in practice, such an overall effect resembles a lossy channel, for simulation purposes, we only consider scenarios where $T_{\rm eq}\leq \eta_{\rm ch} \leq 1$.} We also assume that the mutual information term, $I_{AB}$, which is an observable in the experiment, is given by 
\begin{align}
\label{mutualinfo_CVQKD}
I_{AB}= \frac{1}{2} \log_2 \frac{V+ \chi_{\rm tot}}{1+ \chi_{\rm tot}},
\end{align}
corresponding to a thermal-loss channel identified by $T_{\rm eq}$ and $\xi$. In \cref{mutualinfo_CVQKD}, $V$ is the variance of the two-mode squeezed vacuum (TMSV) state at the source (in the entanglement-based picture), and $\chi_{\rm tot}=\chi_{\rm line}+\frac{\chi_{\rm Hom}}{\eta_{\rm ch}}$ is the total noise, calculated at the transmitter end, where $\chi_{\rm line}= \frac{1-\eta_{\rm ch}}{\eta_{\rm ch}} + \xi$ and $\chi_{\rm Hom}= \frac{1-\eta_{\rm d}}{\eta_{\rm d}} + \frac{\nu_{\rm el}}{\eta_{\rm d}} $ are, respectively, the noise terms due to the channel and the homodyne receiver. Here, $\nu_{\rm el}$ denotes the receiver's electronic noise.

In the following, we obtain a lower bound on the secret key generation rate under above assumptions for the setup in \cref{fig:cv_telmodel} in RR and DR cases, and compare it with the corresponding upper bounds that can be obtained from \cref{Thm:DR-QKD}.

\subsection{Reverse Reconciliation}
\label{sec:RR}
Reverse reconciliation is typically the default choice for CV-QKD systems in highly lossy channels. We first consider this case under the restricted Eve scenario of $\eta_{\rm AE} \leq 1$ while $\eta_{\rm EB} = 1$ in \cref{fig:cv_telmodel}. The key question we would like to explore is how the achievable key rate in the setup with a bypass channel compares with the upper bound that can be obtained from the setup of \cref{fig:scenarios}(b). Interestingly, we find that, under the assumptions outlined above, the two are numerically very close to each other in certain practical regimes of interest.

Let us first explain the limitations we have considered in the special setup shown in \cref{fig:cv_telmodel}. 
Given that this is a linear channel, and our encoding is Gaussian, a Gaussian attack is expected to be the optimal collective attack by Eve. In principle, for any given values of $\eta_{\rm AE}$, $\eta_{\rm S}$ and $\eta_{\rm T}$, there could be a Gaussian attack by Eve that is compatible with {the observed values for total transmissivity $T_{\rm eq}$ and the total equivalent excess noise $\xi$ at the transmitter end}. The Gaussian operation by Eve could take different forms. Here, we only focus on one particular form of attack, which can be modelled by the conventional entangling cloner setup as shown in \cref{fig:cv_telmodel}. Here, Eve combines a TMSV state with variance $V_{\rm E}$, at a beam splitter with transmissivity $\eta_{\rm E}$, with the signal she receives from Alice. The implicit assumption here is that Eve's channel is lossy corresponding to the condition that $\eta_{\rm E} \leq 1$. The conclusions we draw in this section will then only be valid for this type of attack.

In Appendix~\ref{app:CVsetup}, we have calculated the corresponding CM for all parties in \cref{fig:cv_telmodel}, from which the expected values for 
our key observables, $T_{\rm eq}$ and $\xi$ are obtained and, respectively, given by \cref{eq:Teq} and \cref{eq:Xieq}. In the following, { in order to focus on the impact of the restrictions imposed on Eve}, we assume that the receiver has no loss, i.e., $\eta_{\rm d} = 1$, and no electronic noise, i.e., $\nu_{\rm el} = 0$. For any given values of $\eta_{\rm AE}$, $\eta_{\rm S}$ and $\eta_{\rm T}$, we can then find the corresponding values for $\eta_{\rm E}$ and $V_{E}$ that are compatible with observed values of $T_{\rm eq}$ and $\xi$. For the sake of our simulation, we assume that the resulting $\eta_{\rm E}$ is less than or equal to one, to be compatible with the entangling cloner attack considered here.

In order to calculate the key rate for the setup of \cref{fig:cv_telmodel}, we use the CM given in \cref{CM:ABE}, from which all relevant terms can be calculated. $I_{AB}$ is already given by \cref{mutualinfo_CVQKD}. To calculate the Holevo information term, we have
\begin{align}
\chi_{BE}= H(EE') - H(EE'|B), 
\end{align}
where $H(EE')$ and $H(EE'|B)$ can, respectively, be obtained from the corresponding symplectic eigenvalues of the CM for $EE'$ and $EE'|B$; see \cref{fig:cv_telmodel} for notations. The former, $\textbf{V}_{EE'}$, is specified by tracing out modes $A$ and $B$ in the CM of Eq.~\eqref{CM:ABE}. We then numerically find its symplectic eigenvalues, which we denote by $\Lambda_1$ and $\Lambda_2$. The latter CM, $ \textbf{V}_{EE'|B}$, can also be obtained by applying a homodyne measurement on mode $B$: 
\begin{align}
\textbf{V}_{EE'|B}= \textbf{V}_{EE'} - \frac{1}{V_B}\Sigma_{BEE'} \Pi \Sigma_{BEE'}^T, 
\end{align}
where $\Sigma_{BEE'}^T=\left[ C_{BE} {\mathbb Z} ~~~ C_{BE'}\mathbbm{1} \right]$ and $\Pi = {\rm diag}(1,0)$, {with $\mathbb Z={\rm diag}\{1,-1\}$ and $\mathbbm{1}$ being the identity matrix of dimension two} \cite{Weedbrook:GaussQI2012}. In the above, $V_B$, $C_{BE}$, and $C_{BE'}$ are defined in \cref{eq:CM-param}. Denoting the symplectic eigenvalues of $\textbf{V}_{EE'|B}$ by $\Lambda_3$ and $\Lambda_4$, the Holevo information term in the RR case is given by
\begin{align}
\chi_{BE}=g(\Lambda_1)+g(\Lambda_2)-g(\Lambda_3)-g(\Lambda_4), 
\end{align}
where $g(x)= (\frac{x+1}{2}) \log_2 (\frac{x+1}{2}) - (\frac{x-1}{2}) \log_2 (\frac{x-1}{2})$. Note that, in the above calculations, we account for the fact that the state corresponding to $ABEE'$ is not a pure state. This prevents us from calculating all the terms from the CM of $A$ and $B$, as it is common in the unrestricted case.

Let us now fix the observed values for $T_{\rm eq}$ and $\xi_{\rm eq}$ and compare the achievable secret key rates in \cref{fig:cv_telmodel} with the corresponding scenario where the bypass channel is removed, or, equivalently, when $\eta_{\rm S} = 0$. In both cases, some optimization needs to be done to find the lower bound on the key rate. In \cref{fig:cv_telmodel}, while the telescope is part of Bob's secure station, it is not clear how this parameter can be characterised. For any key rate analysis, one should then consider the space of feasible values of $\eta_{\rm S}$ and $\eta_{\rm T}$ and go with the worst case possible. In \cref{fig:cv_telmodel}, this corresponds to going over all possible values of $\eta_{\rm S}$ and $\eta_{\rm T}$ that are compatible with $T_{\rm eq}$ and $\xi_{\rm eq}$, and then find $K_{\rm RR}^{\rm (a)} \equiv \min_{\eta_{\rm S},\eta_{\rm T}} \{K_{\rm RR}\}$. Similarly, for the extended Alice model, we can set $\eta_{\rm S}=0$, and optimize over $\eta_{\rm T}$. For a fixed loss in the link, the higher $\eta_{\rm T}$, the more control is given to Eve. The minimum guaranteed key rate in this case is then given by $K_{\rm RR}^{\rm (b)} \equiv K_{\rm RR}(\eta_{\rm S}=0, \eta_{\rm T}=1)$. We can then compare $K_{\rm RR}^{\rm (a)}$ with $K_{\rm RR}^{\rm (b)}$.  

\begin{figure}[tb]
\includegraphics[width=0.5\textwidth-15pt]{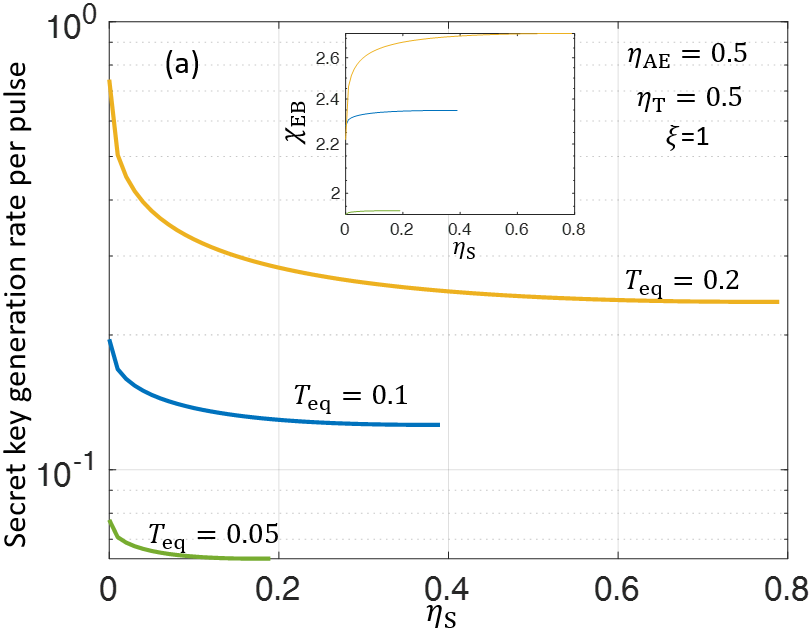}
\includegraphics[width=0.5\textwidth-15pt]{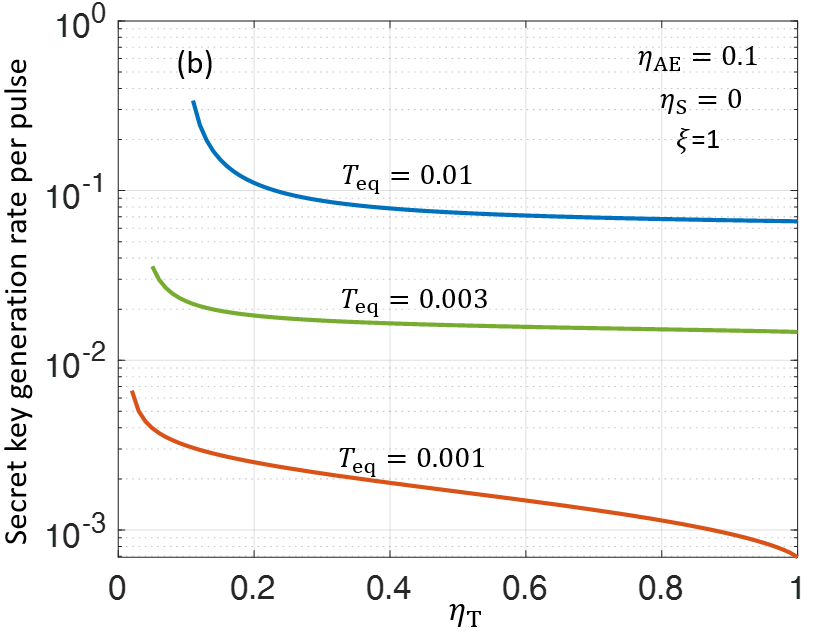}
\caption{Secret key generation rate, for the CV-QKD protocol in \cref{fig:cv_telmodel} with reverse reconciliation, versus (a) $\eta_{\rm S}$ and (b) $\eta_{\rm T}$ { for different values of observed transmissivity and a fixed value of excess noise}. In (a), the bound on the leaked information to Eve, $\chi_{EB}$, is also shown in the inset. In both figures, $V=300$ in SNU, $\beta=1$, $\eta_{\rm d} = 1$, and $\nu_{\rm el} = 0$. Other parameters are specified on the plot.} 
\label{fig:cv_telmodelR}
\end{figure}

In order to get some insight into our optimization problem, in \cref{fig:cv_telmodelR}, we have plotted $K_{\rm RR}$ versus each of $\eta_{\rm S}$ and $\eta_{\rm T}$, while keeping the other parameter constant. { To mainly focus on the impact of the channel parameters in \cref{fig:cv_telmodel}, we have assumed $\beta =1$, which results in optimal $V$ to be very large. We have fixed $V$ at 300 in SNU, which gives us close to optimum key rate values.} In \cref{fig:cv_telmodelR}(a), $\eta_{\rm T}$ and $\eta_{\rm AE}$ are fixed at 0.5, while, for different values of $T_{\rm eq}$, we look at how $K_{\rm RR}$ varies versus $\eta_{\rm S}$. We observe a decreasing behavior for the key rate within the acceptable range of values for $\eta_{\rm S}$. Note that, within the assumptions in our model, e.g. that $0 \leq \eta_{\rm E} \leq 1$, such a range becomes narrower with decrease in $T_{\rm eq}$. This is because in \cref{eq:Teq}, the maximum value for $\eta_{\rm S}$ is given by
{ $\eta_{\rm S}^{\rm max} = T_{\rm eq}/[(1-\eta_{\rm AE})(1-\eta_{\rm T})]$} at $\eta_{\rm E} = 0$, i.e., when Alice's signal reaches Bob only via the bypass channel. Interestingly, at such a point, the key rate is minimum, while $\chi_{EB}$, shown in the inset, is maximum. A justification for this behavior is that, at $\eta_{\rm E} = 0$, Eve can keep the entirety of the signal she has received from Alice for herself, and use it to obtain information about Bob's key. In fact, in this scenario, the bypass channel helps Eve with masquerading the transmissivity of the channel without requiring her to give up any information she can extract from her share of Alice's signal. This observation also explains why the scenario with no bypass channels offers an upper bound on the key rate. In the latter case, i.e., when $\eta_{\rm S} = 0$, we see a similar behavior with regard to the optimum value of $\eta_{\rm E}$ from Eve's perspective. As shown in \cref{fig:cv_telmodelR}(b), in this case, the key rate goes down with increase in $\eta_{\rm T} \geq T_{\rm eq}$. The larger $\eta_{T}$, the smaller will be $\eta_{\rm E} = T_{\rm eq}/\eta_{\rm T}$, meaning that Eve has more control on the channel. This observation agrees with our earlier definition of $K_{\rm RR}^{\rm (b)}$.

Putting together the points made above, it may seem that the gap between $K_{\rm RR}^{\rm (a)}$ and $K_{\rm RR}^{\rm (b)}$ could be large in certain regimes of operation. In \cref{fig:cv_telmodelR}(a), it is, however, interesting to see that the difference between the maximum value of $K_{\rm RR}$ at $\eta_{\rm S} = 0$, and its minimum value, obtained at $\eta_{\rm S}^{\rm max}$, shrinks down as $T_{\rm eq}$ decreases. This would give us the hope that, in practical regimes of operation for satellite QKD with a total loss of 30-40 dB, the difference between $K_{\rm RR}^{\rm (a)}$ and $K_{\rm RR}^{\rm (b)}$ could be reasonably low. This has been verified, as a function of $\eta_{\rm AE}$, in \cref{fig:cv_ExtAlice}(a) at $T_{\rm eq} = 0.001$ for different values of excess noise. As can be seen, $K_{\rm RR}^{\rm (a)}$ and $K_{\rm RR}^{\rm (b)}$ almost overlap in the entire region with the exception of when $\eta_{\rm AE} \ll 1$. Numerically speaking, the optimum value for $K_{\rm RR}^{\rm (a)}$ is often obtained at $\eta_{\rm S} = 1$, which effectively maximises $\eta_{\rm T}$ and minimises $\eta_{\rm E}$. The latter two favour Eve, while the former makes the bypass channel a reliable replacement for what Eve should have done in the absence of the bypass channel. This also suggests that, while our model in \cref{fig:cv_telmodel} is just a special case of what could happen in reality, a no-loss, and possibly no-noise, bypass channel, as we are dealing with in the case of $K_{\rm RR}^{\rm (a)}$, could be the worst case scenario for Alice and Bob. We have briefly examined this hypothesis by considering a thermal-loss bypass channel, and observed the following:
\begin{itemize}
\item The key change in the CM elements is for the excess noise expression in \cref{eq:Xieq}, which now gets an additional term $(1-\eta_{\rm S})(1-\eta_{\rm T})(V_{\rm S} - 1)$, due to the bypass channel, where $V_{\rm S}$ is the variance of the TMSV state that models thermal noise in the bypass channel.
\item At $\eta_{\rm S} < 1$ and $V_{\rm S} > 1$, we see an increase in the key rate as compared to the case of $V_{\rm S} = 1$, corresponding to no thermal noise in the bypass channel. 
\item The minimum key rate is, however, still obtained at $\eta_{\rm S} = 1$, in which case the effect of additional term in the excess noise vanishes, and we will obtain the same result for $K_{\rm RR}^{\rm (a)}$ as the pure-loss bypass channel.
\end{itemize}
We should note that we still limit our search space to the feasibility assumptions we have made in \cref{fig:cv_telmodel}. While the above claim needs to be analytically verified, based on our numerical results, in practical regimes of operation for satellite-QKD, it seems safe to use the upper bound given by \cref{Thm:DR-QKD} as a reliable approximate to the lower bound on the key rate for CV-QKD systems with reverse reconciliation. 

Another reassuring result in \cref{fig:cv_ExtAlice}(a) is that the achievable key rate is a decreasing function of $\eta_{\rm AE}$, that is, the more restriction we set on Eve, the higher key rate Alice and Bob can securely achieve. The impact in certain cases can be quite instrumental. For instance, at a total equivalent excess noise of $\xi = 0.1$ at the transmitter end, while no key can be exchanged under unrestricted Eve, positive key rates can be obtained for $\eta_{\rm AE} < 0.9$. The same happens for $\xi = 1$, but with higher restrictions on Eve at $\eta_{\rm AE} < 0.1$. Interestingly, when $\eta_{\rm AE}$ is sufficiently low, the key rate will become almost independent of the amount of excess noise, and rather large key rates can be obtained.

\begin{figure}[t]
\includegraphics[width=0.5\textwidth-15pt]{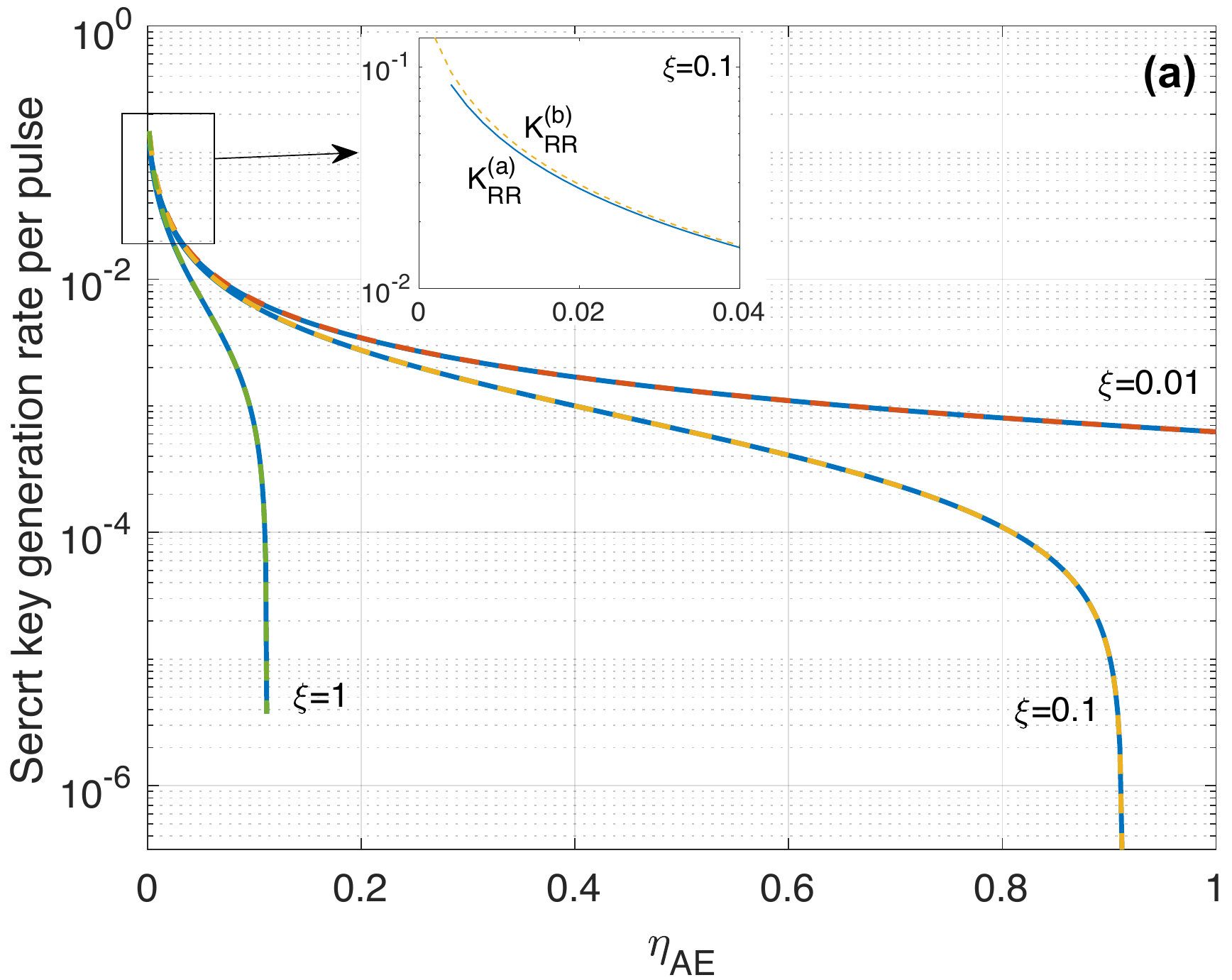}
\includegraphics[width=0.5\textwidth-15pt]{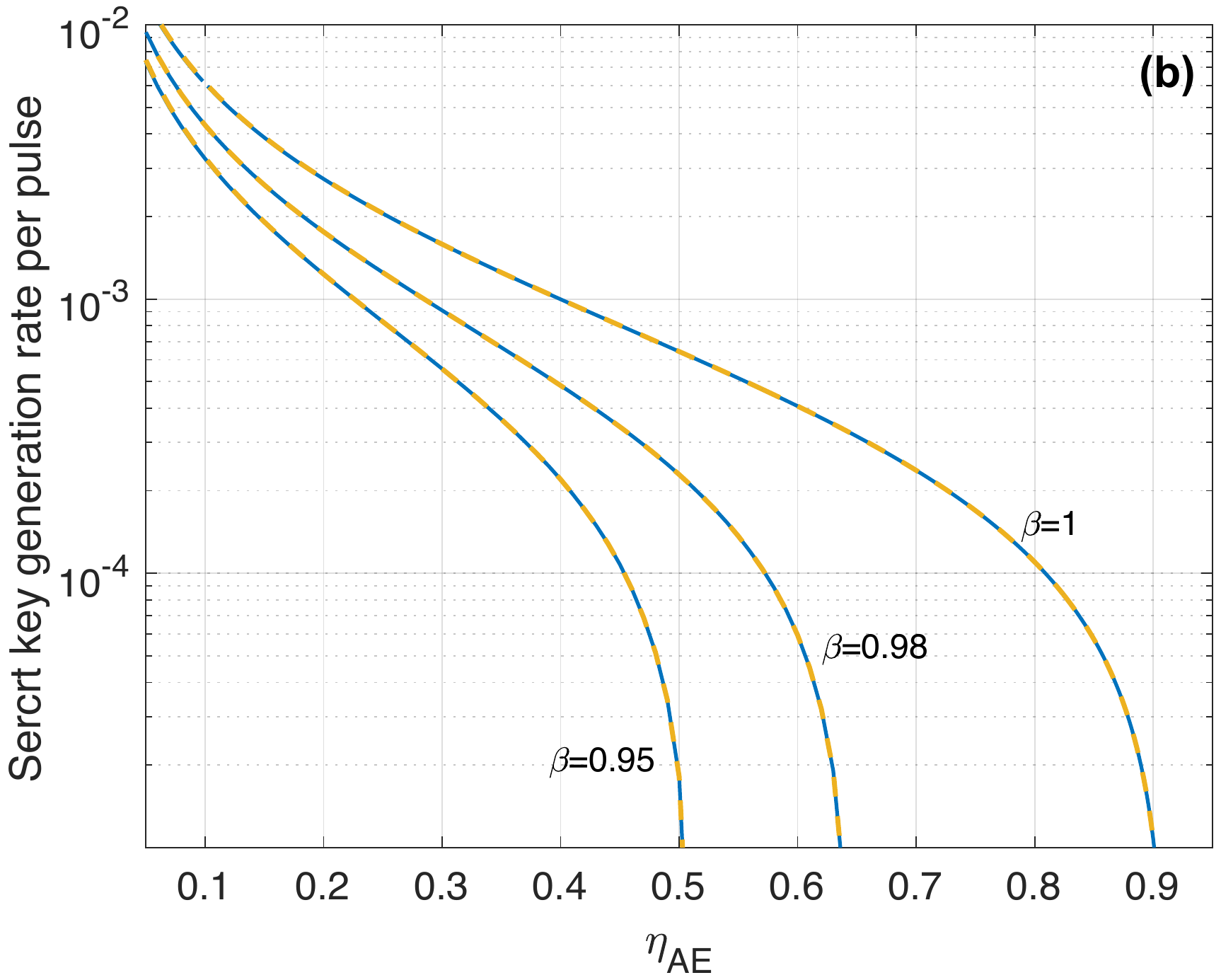}
\caption{Secret key generation rate, for the CV-QKD protocol in \cref{fig:cv_telmodel} with reverse reconciliation, versus $\eta_{\rm AE}$. The solid lines represent $K_{\rm RR}^{\rm (a)}$ and the dashed lines represent $K_{\rm RR}^{\rm (b)}$. As shown in the inset, the two curves are very close to each other, and mostly overlap except for small values of $\eta_{\rm AE}$. In {\bf (a)}, $T_{\rm eq} = 0.001$, $V=300$ in SNU, $\beta=1$, $\eta_{\rm d} = 1$,  $\nu_{\rm el} = 0$, and the excess noise is shown on the graphs. 
{ In {\bf (b)}, we consider imperfect reconciliation efficiencies characterized by parameter $\beta$. Other parameters are $T_{\rm eq} = 0.001$, $\xi=0.1$, and $V=3.5$ in SNU.}
} 
\label{fig:cv_ExtAlice}
\end{figure}

{ The overall results explained above seem to be unchanging when we account for other sources of imperfection in our system. In particular, in \cref{fig:cv_ExtAlice}(b), we have accounted for non-ideal values for the reconciliation efficiency parameter $\beta$. It can be seen that the overlap between the upper and lower bounds on the key rate still holds when $\beta < 1$, and that the key rate goes down as $\eta_{\rm AE}$ increases. The difference is that the threshold value for $\eta_{\rm AE}$ to give us positive key rates goes down as we decrease $\beta$. This is understandable because, by reducing the mutual information term by a factor of $\beta$, we now need further restrictions on Eve to bring down the Holevo information term in \cref{rate_RR}. The transition to positive key rates happen at around 0.5 for $\eta_{\rm AE}$  at $\beta =0.95$, which is still an attainable value.}

\subsection{CV-QKD with Direct Reconciliation}
In the previous section, we saw how the proposed restrictions on Eve can improve the key rate of CV-QKD systems in highly lossy channels. Here, we apply the results of \cref{Sec:SecProof} to the case of CV-QKD with DR under a restricted Eve. In the DR case, with no restriction on Eve, the maximum loss that we can tolerate is only 3~dB. It would be interesting to see how that would change when we impose restrictions on Eve's access to Alice's signal. In the following, we consider two extremes: when $\eta_{\rm AE} > T_{\rm eq}$, in which case, the entangling cloner attack as in \cref{fig:cv_telmodel} is the optimal attack by Eve, and when $\eta_{\rm AE} < T_{\rm eq}$, where we can use \cref{eq:obv-bound} to directly find a lower bound on the key rate.

\subsubsection{Method 1: Entangling Cloner Attack}
Here, we assume that $T_{\rm eq} < \eta_{\rm AE}$, and use the results of Appendix \ref{app:CVsetup} to calculate the key rate for the setup of \cref{fig:cv_telmodel}. As in the RR case, we optimize the key rate over uncharacterized system parameters $\eta_{\rm S}$ and $\eta_{\rm T}$ as follows:
\begin{eqnarray}
&K_{\rm DR}^{\rm (a)} \equiv \min_{\eta_{\rm S},\eta_{\rm T}} \{K_{\rm DR}(\eta_{\rm S},\eta_{\rm T})\} & \nonumber \\
&K_{\rm DR}^{\rm (b)} \equiv \min_{\eta_{\rm T}} \{K_{\rm DR}(0,\eta_{\rm T})\}= K_{\rm DR}(0,1),&
\end{eqnarray}
where $K_{\rm DR}$ is defined in \cref{rate_DR}, with 
\begin{align}
\label{eq:DR1}
\chi_{AE}= H(EE') - H(EE'|A_x),
\end{align}
where $A_x$ represents the homodyne measurement result on one of the quadratures of mode $A$ after going through the 50:50 beam splitter in the heterodyne measurement $M_A$. The above entropy terms can be calculated using the CM in \cref{CM:ABE} with some modifications due to the 50:50 beam splitter in $M_A$. The joint CM for modes $A_x E E'$ is then given by 
\begin{align}
\textbf{V}_{A_xEE'}= \left(\begin{array}{ccc}
\label{CM:AEE}
(V+1)/2 \mathbbm{1} &   0 \mathbbm{1}  & C_{AE'}/\sqrt{2} {\mathbb Z}   \\
0  \mathbbm{1}  &   V_E \mathbbm{1} & C_{EE'}  {\mathbb Z}   \\
C_{AE'}/\sqrt{2} {\mathbb Z}  &  C_{EE'} {\mathbb Z}  & V_{E'}    \mathbbm{1} 
\end{array}\right),
\end{align}
where $\mathbb Z={\rm diag}\{1,-1\}$, $\mathbbm{1}$ is the identity matrix of dimension two, and all other parameters are given by \cref{eq:CM-param}.
Eve's state $\rho_{EE'}$ is then described by the CM $\textbf{V}_{EE'}$, which is given by the $4\times4$ submatrix in the lower right of $\textbf{V}_{A_xEE'}$ given in Eq.~\eqref{CM:AEE}. We then have
\begin{align}
H(EE') = g(\Lambda_1) + g(\Lambda_2), 
\end{align}
where $\Lambda_1$ and $\Lambda_2$ are the symplectic eigenvalues of $\textbf{V}_{EE'}$. Similarly, the conditional term $H(EE'|A_x) = g(\Lambda_3) + g(\Lambda_4)$,  where $\Lambda_3$ and $\Lambda_4$ are the symplectic eigenvalues of $\textbf{V}_{EE'|A_x}$, given by 
\begin{align}
\textbf{V}_{EE'|A_x}= \textbf{V}_{EE'} - \frac{2}{V+1}\Sigma_{A_xEE'} \Pi \Sigma_{A_xEE'}^T, 
\end{align}
where we have applied a homodyne measurement on mode $A_x$ \cite{Weedbrook:GaussQI2012} and 
$\Sigma_{A_xEE'}=\left(\begin{array}{cc}
0  \mathbbm{1}   \\
C_{AE'}/\sqrt{2}{\mathbb Z}
\end{array}\right)$.

\subsubsection{Method 2: Generic Lower Bound}
In Method 2, we use \cref{eq:obv-bound}, which basically uses the state before Eve's operation, to bound $\chi_{AE}$. The advantage of this technique is that here we do not need to impose any conditions on the observed values of $\eta_{\rm ch}$ and $\eta_{\rm AE}$. In particular, we can now cover the case of $\eta_{\rm AE}<\eta_{\rm ch}$, which is the extreme case where Eve's collection efficiency is worse than Bob, for instance, as in \cref{fig:setup}(c). In this case, we use \cref{eq:obv-bound} to upper bound $\chi_{AE}$ by 
\begin{align}
\chi_{AB'}= H(B') - H(B'|A_x),
\end{align}
where $B'$ is mode $B$ right after the first beam splitter in \cref{fig:scenarios}(a). Note that, in this approach, we do not need to restrict ourselves to the assumptions in \cref{fig:cv_telmodel}. 
In the above equation, $H(B')$ is the von Neumann entropy of the thermal state $B'$ with variance $V_{B'}=\eta_{\rm AE}V+1-\eta_{\rm AE}$. We then use the fact that the symplectic eigenvalue of a single-mode thermal state is indeed equal to its variance to obtain $H(B') = g(V_{B'})$. Similarly, to calculate the term $H({B'|A_x})$, we need to find the symplectic eigenvalues for the conditional covariance matrix $\textbf{V}_{B'|A_x}$. Given that the CM of $A_xB'$ is given by
\begin{align}
\textbf{V}_{A_xB'}= \left(\begin{array}{cc}
(V+1)/2 \mathbbm{1}     &  \sqrt{\eta_{\rm AE}(V^2-1)/2} {\mathbb Z} \\
\sqrt{\eta_{\rm AE}(V^2-1)/2} {\mathbb Z} &  V_{B'} \mathbbm{1}    \\
\end{array}\right), 
\end{align}
we have, after the homodyne detection on $A_x$,
\begin{align}
\textbf{V}_{B'|A_x}&=  V_{B'} \mathbbm{1} - \frac{\eta_{\rm AE}(V^2-1)}{V+1} {\mathbb Z}\Pi{\mathbb Z}^T \nonumber \\
& = \left(\begin{array}{cc}
1    &  0 \\
0 &  V_{B'}    \\
\end{array}\right) .
\end{align}
An upper bound on $\chi_{AE}$ can then be calculated from the following 
\begin{align}
\label{eq:DR-M2}
\chi_{AE}\leq g(V_{B'})-  g(\sqrt{V_{B'}}). 
\end{align}

\subsubsection{Numerical Results}
Figure~\ref{fig:rateDR_M1M2}(a) shows the key rate versus $\eta_{\rm AE}$, for $\eta_{\rm AE} > T_{\rm eq}$, using Method 1 for different values of $T_{\rm eq}>0.5$. We have plotted the upper bound $K_{\rm DR}^{(b)}$ (dashed lines) as well as the optimised lower bound $K_{\rm DR}^{(a)}$ (solid lines). Unlike the RR case, in the DR scenario, the two bounds are not close and effectively we cannot guarantee higher key rates than what we can obtain in the unrestricted case. In particular, for $T_{\rm eq}<0.5$, similar to the unrestricted case, we do not get a positive key rate for $K_{\rm DR}^{(a)}$. The optimum value of $K_{\rm DR}^{(a)}$ is again numerically obtained at $\eta_{\rm S} = 1$, but this time optimum $\eta_{\rm E}$ takes rather large nonzero values around 0.5. The larger $\eta_{\rm AE}$ is, the larger $\eta_{\rm E}$ we get at the optimum point. This could be because, at $\eta_{\rm AE}$ close to one, the main path through Eve should offer a transmissivity close to 0.5, or higher, to get positive key rates, whereas, as $\eta_{\rm AE}$ goes down, the bypass channel helps Eve more with the total observed $T_{\rm eq}$ to the extent that the initial restriction on Eve becomes irrelevant.

\begin{figure}[t]
\includegraphics[width=0.5\textwidth-15pt]{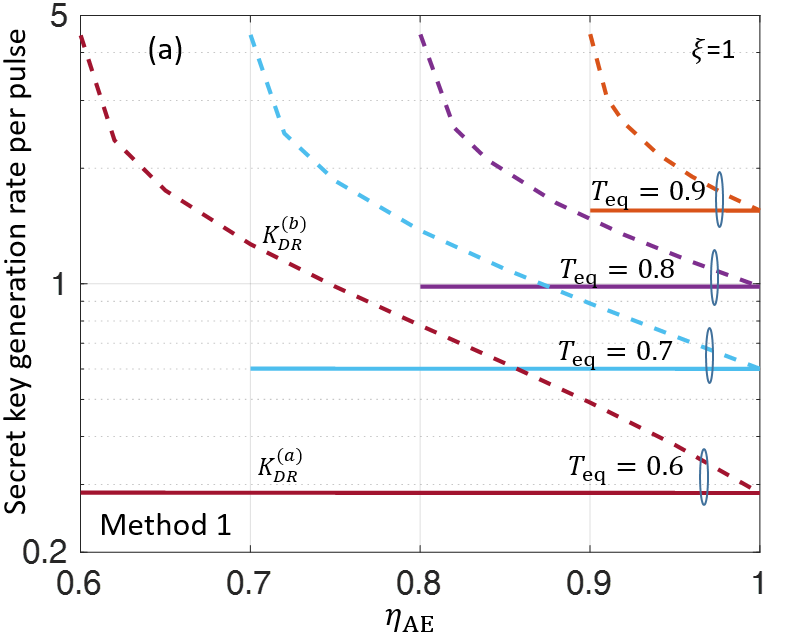}
\includegraphics[width=0.5\textwidth-15pt]{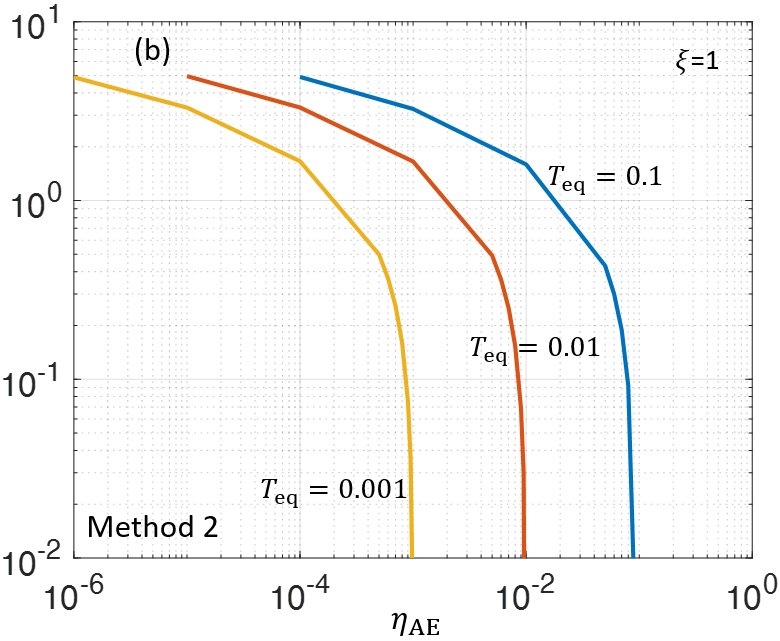}
\caption{Secret key generation rate, at $V= 10^7$ in shot-noise unit (SNU), obtained from (a) Method 1 and (b) Method 2, versus $\eta_{\rm AE}$ for CV-QKD systems using direct reconciliation. The results are shown for an observed channel with different values of $T_{\rm eq}$, $\xi=1$ SNU, $\eta_{\rm d} = 1$, $\nu_{\rm el} = 0$, and $\beta=1$. In (a), the solid (dashed) lines represent $K_{\rm DR}^{(a)}$ ($K_{\rm DR}^{(b)}$).} 	\label{fig:rateDR_M1M2}
\end{figure}

We can however get some advantage in the restricted case in the extreme case of $\eta_{\rm AE} < T_{\rm eq}$. Here, we can use the generic lower bound in \cref{eq:DR-M2} to obtain the key rate. The results are shown in \cref{fig:rateDR_M1M2}(b). As can be seen, in this case, the key rate can improve by orders of magnitude by decreasing $\eta_{\rm AE}$. The seemingly flat curves at the left-hand-side of the graph are mainly because of the choice of a finite value for $V$. In principle, the key rate would continue going up in the asymptotic limit of $V \rightarrow \infty$. However, the growth happens very slowly, e.g., for a variance as large as $V=10^{20}$, the key rate is only about 25. Considering the limitations on the transmitted power, a maximum $V$ can be chosen in practice to offer the maximum key rate in such settings where Eve is disadvantaged as compared to Bob, as in the case of the wiretap channel.

\section{Discrete-variable protocols with restricted Eve}
\label{dv-qkd}

In this section, we consider several DV-QKD protocols, mainly focusing on the BB84 protocol \cite{bennett2014quantum}, and its variants. We consider the original BB84 with single-photon sources (SPSs) as well as its variant with phase randomized weak coherent pulses (WCPs) \cite{gottesman2004security}. In all these cases we deal with a photon-number channel from Eve's perspective. We assume $\eta_{\rm EB}=1$, i.e., we only consider Eve's restriction on her signal collection capabilities. The case of $\eta_{\rm EB} < 1$ will be the subject of another investigation. In the following, we present a method to obtain a lower bound for the secret key rate in the restricted Eve case. In this paper, we only consider the asymptotic regime where infinitely many signals are exchanged and focus on how restrictions on Eve can affect system performance.

\subsection{General Lower Bounds for Secret Key Rate}
\label{sec:DV-Methods}
The secret key rate of BB84 protocols, in the asymptotic regime, in an unrestricted Eve scenario is lower bounded by \cite{scarani2009security}
\begin{align}
R \geq q Q \left[ -f h(E)+\frac{Q_1}{Q}(1-h(e_1))+\frac{Q_0}{Q}\right],
\label{general_formula_DV}
\end{align}
where $f$ is the error correction inefficiency, $q$ is the basis reconciliation factor, and $h(\cdot)$ represents Shannon's binary entropy function defined as
\begin{align}
h(x)=-x\log_2x-(1-x) \log_2 (1-x).
\end{align}
In Eq.~(\ref{general_formula_DV}), $Q$, $E$, and $e_1$, respectively, denote the total gain, QBER, and single-photon error rate. The parameters $Q_0$ and $Q_1$ are given by
\begin{eqnarray}
Q_0=Y_0 p_0,\nonumber\\
Q_1=Y_1 p_1,
\end{eqnarray}
where $Y_i$ is the probability of Bob's detection under the condition that Alice has sent $i$ photons, and $p_i$ denotes the probability that Alice sends $i$ photons. 

The general idea behind \cref{general_formula_DV} is that, in photon-number channels, the information gained by Eve depends on the number of photons in the signal received by Eve. For the events in which Eve receives two or more photons, one may assume that Eve can obtain full information about the transmitted key bit using the photon-number splitting (PNS) attack \cite{brassard2000limitations}. In the events in which Eve receives one photon, the maximum information that she can gain is $h(e_1)$. Finally, if Eve receives no photon, her information is zero, assuming that direct reconciliation is used.

In our restricted Eve scenario, for every sifted bit, we find an upper bound, $I_{\rm E}$, on Eve's information, in the direct reconciliation case, based on the number of photons transmitted by Alice and received by Eve, denoted, respectively, by $n$ and $m$, as follows
\begin{equation}
 I_{\rm E}=\begin{cases} 
  0 & m=0, n\geq 0 \\
  1 & m>1, n\geq m \\
  h(\varepsilon_{11}) & m=1, n=1\\
  1 & m=1, n >1 
\end{cases}
\label{eq:DVcases}
\end{equation} 
where $\varepsilon_{11}$ denotes an upper bound on the error rate of the signals for which $n=m=1$. Here, we have pessimistically assumed that Eve can distinguish between the cases where $m=n=1$ versus $m=1$, but $n>1$. This assumption would allow Eve, in the latter case, to keep the photon to herself and wait to see if one of the remaining photons is received by Bob. To find a lower bound on the secret key rate, we define the parameters $W_{ij}$ and $p_{ij}$ as follows:
\begin{eqnarray}
&W_{ij}=\Pr({\rm Bob's\; detection}|n=i,m=j),& \nonumber \\
&p_{ij}=\Pr(n=i,m=j).&
\end{eqnarray}

In the asymptotic case where Alice sends infinitely many signals, a lower bound on the secret key rate can be obtained by
\begin{align}
R \geq q Q \left[ -f h(E)+\frac{S_{11}}{Q}(1-h(\varepsilon_{11}))+\frac{S_0}{Q}\right],
\label{general_DV_res_b1-no-bound}
\end{align}
where 
\begin{eqnarray}
&S_0=\sum_{i=0}^{\infty}{W_{i0}p_{i0}},&\nonumber\\
&S_{11}=W_{11}p_{11}.&
\label{s_0s_1eq}
\end{eqnarray}
Effectively, the last two terms in \cref{general_DV_res_b1-no-bound} have replaced that of \cref{general_formula_DV}, in the case of no bypass channel, and represent part of the shared key that can be used for privacy amplification. 

In the following, we find bounds on the key parameters in \cref{general_DV_res_b1-no-bound}. In a typical QKD protocol, it may not be possible to measure the exact values of $S_0$, $S_{11}$ and $\varepsilon_{11}$. Instead, we try to find lower bounds on $S_0$ and $S_{11}$, and an upper bound on $\varepsilon_{11}$. To find a lower bound on $S_0$ and $S_{11}$, in the first step we find a lower bound on $S_0+S_{11}$. Note that
\begin{align}
1 \geq Q=\sum_{i=0}^{\infty}{\sum_{j=0}^{i}{W_{ij}p_{ij}}}
=S_0 + S_{11} + S_{\rm other} ,
\label{eq:Q}
\end{align}
where
\begin{align}
S_{\rm other}=\sum_{j=2}^{\infty}{\sum_{i=j}^{\infty}{W_{ij}p_{ij}}}+\sum_{i=2}^{\infty}{W_{i1}p_{i1}} .
\label{S_other}
\end{align}
Using Eqs.~\eqref{s_0s_1eq}--\eqref{S_other}, we can obtain
\begin{align}
S_0+S_{11}= & Q-S_{\rm other} \geq S_{0+11}^{L} \notag \\
\equiv &  Q- \Bigg(\sum_{j=2}^{\infty}{\sum_{i=j}^{\infty}{p_{ij}}}+\sum_{i=2}^{\infty}{p_{i1}}\Bigg),
\label{S_01^L}
\end{align}
where $S_{0+11}^{L}$ denotes the lower bound on $S_0+S_{11}$. Now, we consider the following two inequalities:
\begin{eqnarray}
S_0 \leq \sum_{i=0}^{\infty}{p_{i0}}, \nonumber\\
S_{11} \leq p_{11},
\label{S0_S_11}
\end{eqnarray}
Note that $\sum_{i=0}^{\infty}{p_{i0}}$ is the probability that Eve receives no photon, i.e., $m=0$. We denote this probability by $p^{\rm Eve}_0$. Then, we can write 
\begin{eqnarray}
S_0 \geq S_{0+11}^{L}-S_{11} \geq S_{0+11}^{L}-p_{11},\nonumber\\
S_{11} \geq S_{0+11}^{L}-S_0 \geq S_{0+11}^{L}-p^{\rm Eve}_0.
\label{eq:S0LS1L}
\end{eqnarray}
Substituting Eq.~(\ref{S_01^L}) into the above inequalities, it can be concluded that
\begin{align}
\label{s_0s_1}
S_0 \geq S_0^L \equiv \max \Big\{Q-(1-p^{\rm Eve}_0),0\Big\},\nonumber\\
S_{11} \geq S_{11}^L \equiv \max\Big\{Q-(1-p_{11}),0\Big\}.
    \end{align}
The above bounds have an easy explanation. Let us look at $S_0^L$, for instance. The term $1-p^{\rm Eve}_0$ is the probability that Eve has got a non-vacuum state. This sets an upper bound on the number of detection events that Bob can get because of non-vacuum states. Any other click must come from cases where Eve has received no photons, which gives us the expression in \cref{s_0s_1}.

Note that, in the case of restricted Eve, the bound on $S_0$ is likely to become relevant for small values of $\eta_{\rm AE}$. This is because, for $S_0^L$ to be strictly positive, $1-p^{\rm Eve}_0$ should be smaller than $Q$. In the nominal mode of operation, when no Eve is present, $Q$ often scales with channel transmissivity, and, for coherent state inputs, $1-p^{\rm Eve}_0$ is expected to scale with $\eta_{\rm AE}$. This suggests that as $\eta_{\rm AE}$ becomes smaller and smaller, there could be a non-negligible contribution from the $S_0$ term, which is often ignored in the conventional unrestricted Eve case. In the latter case, $1-p_0$ is often a fixed value, which $Q$ could easily become smaller than in high-loss regimes. Even if $Q$ happens to be larger than $1-p_0$, the contribution from  $S_0$ is likely to be cancelled out by the additional error correction that Alice and Bob need to do for the clicks resulted from the vacuum states sent by Alice. In the restricted Eve scenario, however, the bypass channel can, in principle, provide a route to obtaining correlated data between Alice and Bob without necessarily increasing the QBER. This could allow Alice and Bob to extract more secret key bits from their measured data as compared to the conventional scenario. We will look more carefully at the effect of the above bounds on $S_0$ and $S_{11}$ later in this section.

To find an upper bound on $\varepsilon_1$, we note that
\begin{align}
EQ=\sum_{i=0}^{\infty}{\sum_{j=0}^{i}{\varepsilon_{ij} S_{ij}}},
\end{align}
where $S_{ij}=W_{ij}p_{ij}$. Using the above equation, we can write
\begin{align}
EQ\geq \varepsilon_{11} S_{11} \geq \varepsilon_{11} S_{11}^L \Rightarrow \varepsilon_{11} \leq \frac{EQ}{S_{11}^L}.
\label{error}
\end{align}
Using \cref{s_0s_1} and \cref{error}, the secret key rate, in the restricted Eve case, in the limit of infinitely long key is lower bounded by
\begin{align}
R\geq q Q \left[ -f h(E)+\frac{S_{11}^{L}}{Q}(1-h(\varepsilon_{11}^{U}))+\frac{S_0^{L}}{Q}\right],
\label{general_DV_res_b1}
\end{align}
where $\varepsilon_{11}^{U}=\min\{EQ/S_{11}^L,1/2\}$ gives an upper bound on $h(\varepsilon_{11})$. 

\subsection{BB84 Performance Under Restricted Eavesdropping}
\label{sec:SPS}
In the following, we discuss the secret key rate of BB84 protocols considering different sources. We find the relevant parameters needed in each case to calculate $R$ as given by \cref{general_DV_res_b1}.

\subsubsection{BB84 with Single-Photon Sources}
If an ideal single-photon source is used at Alice side, we have $S_0+S_{11}=Q$, $p^{\rm Eve}_0=1-\eta_{\rm AE}$, and $p_{11}=\eta_{\rm AE}$. Hence, from \cref{s_0s_1}, we have 
\begin{eqnarray}
S^L_0=\max\Big\{Q-\eta_{\rm AE},0\Big\},\nonumber\\
S^L_{11}=\max\big\{Q-(1-\eta_{\rm AE}),0\Big\}.
\label{single_S_0S_1}
\end{eqnarray}
By substituting Eq.~(\ref{single_S_0S_1}) into Eqs.~(\ref{error}) and (\ref{general_DV_res_b1}), we can calculate a lower bound on the secret key rate. 

{ In the case of an ideal single-photon source, there are alternative ways of calculating lower bounds on the key rate by directly using \cref{general_DV_res_b1-no-bound} and \cref{general_formula_DV}. For instance, because $S_0+S_{11}=Q$, \cref{general_DV_res_b1-no-bound} turns into
\begin{align}
R &\geq q Q \left[ -f h(E)+1 - \frac{S_{11}}{Q}h(\varepsilon_{11}) \right] \nonumber \\
&\geq q Q \left[ -f h(E)+1 - h(\varepsilon_{11}^U) \right]
\label{general_DV_res_SPS1}
\end{align}
Alternatively, one can directly use \cref{general_formula_DV} by setting $Q_0 = 0$. In the numerical section, we use the best of these three bounds to specify the lower on the key rate.}

\subsubsection{BB84 with WCP Sources}
\label{sec:WCP}
Phase-randomised WCP (or, in short, WCP) sources follow Poisson distribution in photon generation. If the average number of photons of the WCP source is $\mu$, then $p_{ij}$ can be obtained by
\begin{align}
p_{ij}= & Pr(n=i)Pr(m=j|n=i) \notag \\
= & \frac{e^{-\mu} \mu^i}{i!}\binom{i}{j} \eta_{\rm AE}^j (1-\eta_{\rm AE})^{i-j}
\label{WCP_p}
\end{align}
By substituting the above equation into Eq.~(\ref{S0_S_11}) and Eq.~(\ref{s_0s_1}), we obtain
\begin{eqnarray}
&S_0 \geq S^L_0=\max\Big\{Q-(1-e^{-\mu \eta_{\rm AE}}),0\Big\},&\nonumber\\
&S_{11} \geq S^L_{11}=\max\Big\{Q-(1-\mu \eta_{\rm AE} e^{-\mu}),0\Big\}.&
\label{WCP_S_0_S_1}
\end{eqnarray}
The lower bound $R$ can then be obtained by substituting the above two equations into Eq.~(\ref{error}) and Eq.~(\ref{general_DV_res_b1}).

\subsubsection{Numerical Results}
In this subsection, we consider a satellite-based QKD system, using the BB84 protocol, and evaluate its performance in different regimes of operation. Nominal values used for system parameters are listed in Table~\ref{Tab:Para}. Noteworthy is the fact that we calculate the key rate at a channel transmissivity of $\eta_{\rm ch}=10^{-3}$ corresponding to the recent efficiency measurements for the Micius satellite \cite{Entg-based-SatQKD2020}. We have also assumed the ground station is equipped with superconducting single-photon detectors of 90\% efficiency, but to account for possible background noise in the link \cite{Hugo_RestricedEve_PRApplied}, the dark count probability per pulse for the receiver is assumed to be $p_{\rm dc} = 10^{-7}$. For a system running at 100~MHz, this is one order of magnitude higher than the typical dark counts for such detectors \cite{SupDet93}. We also assume that we use the efficient version of the BB84 protocol \cite{lo2005efficient}, in which the reconciliation factor $q$ approaches one. 

\begin{table}[b]
\caption{Nominal values used for system parameters.}
\centering 
\begin{tabular}{|c |c |} 
    \hline
    Parameter & Value\\
    \hline		
    Average channel loss, $\eta_{\rm ch}$& 30 dB\\
    Error correction inefficiency, $f$& 1.16\\
    Basis reconciliation factor, $q$ & 1\\
    Total dark/background probability, $p_{\rm dc}$ & 1E-7 \\
    misalignment error, $e_d$ & 0.01\\
    quantum efficiency of detectors, $\eta _d$ & 0.9\\
    \hline
\end{tabular}
\label{Tab:Para}
\end{table} 

We consider two types of sources: SPS and WCP for the encoder at Alice side, i.e., the satellite. In a real QKD experiment, the parameters related to the overall gain and the QBER, i.e., $Q$ and $E$ in Eq.~(\ref{general_DV_res_b1}), are obtained by measurement. Here, we assume that the measured values for these parameters are equal to the ones that can be obtained analytically as calculated in Appendix A of Ref.~\cite{Panayi_2014}.

\begin{figure}[t]
\includegraphics[width=0.5\textwidth-15pt]{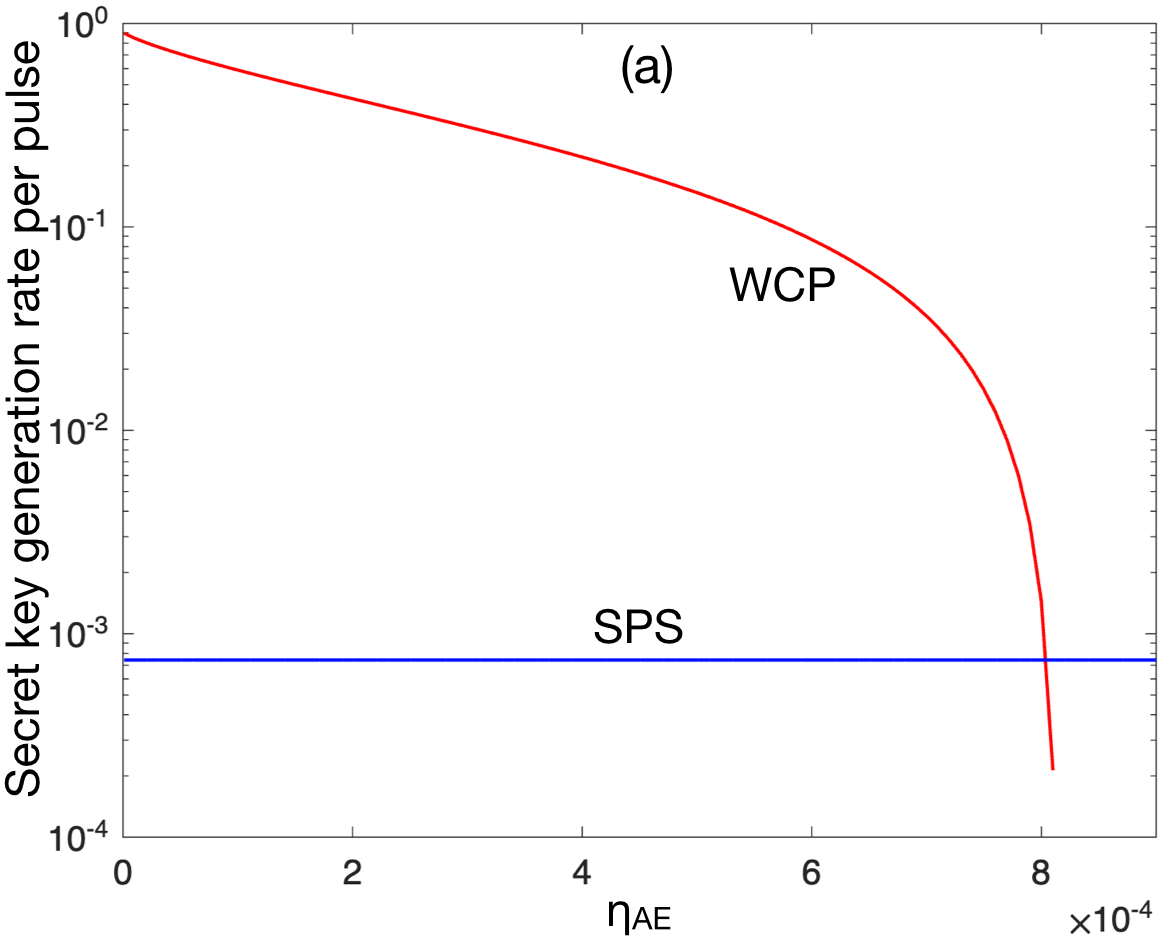}
\includegraphics[width=0.5\textwidth-15pt]{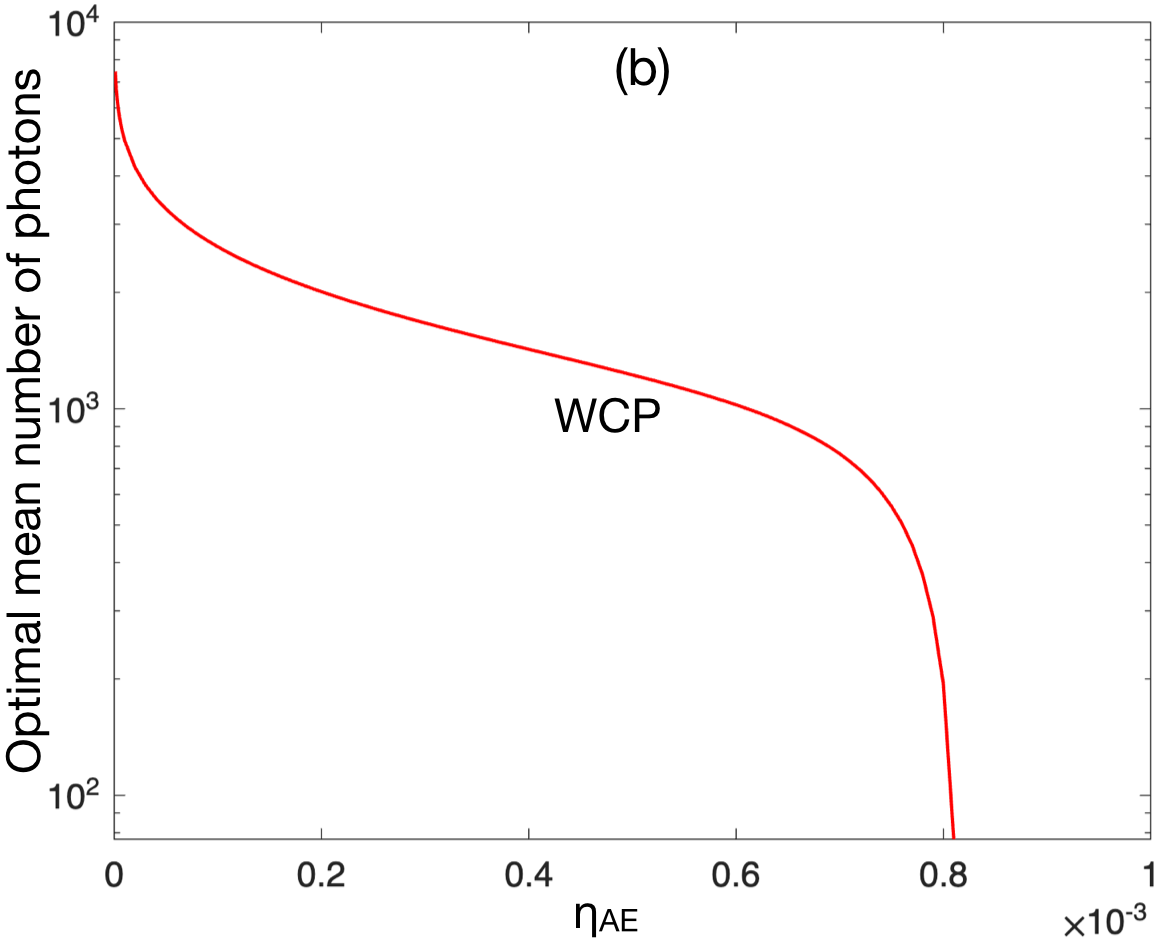}
\caption{(a) Secret key generation rate versus $\eta_{\rm AE}$ for WCP and SPS sources. (b) Optimal values of $\mu$ versus $\eta_{\rm AE}$ for WCP sources.} 
\label{fig:DV-rate}
\end{figure}
Figure~\ref{fig:DV-rate}(a) shows the secret key rate versus $\eta_{\rm AE}$ for SPS and WCP protocols. We have optimized the key rate over $\mu$ in the WCP case. The optimum values of $\mu$ are shown in \cref{fig:DV-rate}(b). There are several interesting points to highlight in \cref{fig:DV-rate}:
\begin{itemize}
\item At the channel loss of 30 dB, the WCP protocol cannot provide any secret key under unrestricted Eve's assumption. In the restricted Eve case, however, we start having positive key rates for roughly $\eta_{\rm AE}< 8.1 \times 10^{-4}$. This suggests a simple phase randomized laser source is sufficient for key exchange in this regime.
\item The WCP protocol performance exceeds that of the SPS protocol at small values of $\eta_{\rm AE}$. This is interesting as the SPS source conventionally corresponds to the ideal BB84 protocol. In our example system, this happens at roughly $\eta_{\rm AE} < 8 \times 10^{-4}$. This is mainly because of the extra laser power that Alice can now use to generate signals with larger number of photons without worrying much about photon number splitting attacks. We do not have this possibility with SPSs, hence such sources would not allow us to benefit from Eve's restrictions in this case. 
\item { Among the three techniques proposed in \cref{sec:SPS} for the SPS source, the one obtained from \cref{general_formula_DV} offers the highest key rate. That is why the corresponding curve in \cref{fig:DV-rate} remains constant. Mathematically, this can be seen by comparing \cref{general_DV_res_SPS1} with \cref{general_formula_DV}, and noting that $h(\varepsilon_{11}^U) \geq h(e_1)$. The worst-case assumption made in \cref{eq:DVcases} seems to not offer any advantage in the single-photon case. To check if there is any room for improvement, we have verified if the bound can be improved by using numerical techniques for bounding the key rate \cite{Winick2018reliablenumerical, Bunandar2020}. We have however observed no change in the achievable rate and the result presented in \cref{fig:DV-rate} seems to be the optimum case for the SPS source.} The full detail of the numerical approach will be the subject of a separate publication. 
\item As mentioned earlier, the case of $\eta_{\rm AE} < \eta=\eta_{\rm ch}\eta_d$ is of special interest. This is when the bound $S_0^L$ in \cref{WCP_S_0_S_1} can take nontrivial values. We can see this effect in the parameter values chosen for our simulation, where $Q = 1-(1-p_{\rm dc})^2e^{-\eta\mu}$. In this case, we have $Q-(1-p_0^{\rm Eve}) = e^{-\mu'}-(1-p_{\rm dc})^2 e^{-\eta\mu}>e^{-\eta_{\rm AE}\mu} - e^{-\eta\mu}$. The latter term would get a positive value when $\eta_{\rm AE} < \eta$, resulting in a positive value for $S_0^L$. 
\end{itemize}

\section{Conclusions and Discussion}
\label{SecCon}
The security of prepare-and-measure QKD systems under certain restrictions on the eavesdropper was studied. We relaxed some of the assumptions on the eavesdropper's unrestricted capabilities in collecting and re-transmitting QKD signals.
Such restrictions could particularly find relevance in satellite-based QKD protocols. Our restrictive assumptions resulted in an under-explored scenario, where the channel between Alice and Bob is not entirely controlled by Eve, but rather an uncharacterized bypass channel could also carry signal. We found generic upper bounds on the key rate for QKD systems in the presence of bypass channels, and in the case of CV-QKD with reverse reconciliation showed that the upper and lower bounds on the key rate are very close to each other in certain practical regimes of interest. Such an upper bound offers a considerable boost to the key rate that can be achieved under unrestricted eavesdropping. In the case of CV-QKD with direct reconciliation, or that of BB84 protocols, the advantage offered by our customized bound was limited to certain scenarios where Eve's access to Alice's signal is significantly hampered, as is the case, for instance, in wiretap channels. Nevertheless, our approach to security proof relies only on a few assumptions, which can, in principle, be verified with monitoring techniques.

The analysis of QKD systems in the presence of bypass channels can certainly be extended in several directions, where each is worth a separate investigation. For instance, the difference between reverse and direct reconciliation in the CV-QKD case raises the question of whether DV-QKD with reverse reconciliation could offer any better performance. One way to answer such questions is by developing numerical techniques for finding tight bounds on the key rate in such setups, which is ongoing research. While Theorem 1 is applicable to finite-size key settings, the issue of statistical fluctuations in the presence of the bypass channel needs to be further investigated. Whether the bypass channel affects non prepare-and-measure QKD protocols, e.g., entanglement-based QKD, also needs to be investigated. { In this work, we mainly focused on LEO satellite scenarios, but, in principle, the same techniques could find application in medium-earth orbit and geostationary satellite missions as well. The practicality of this needs to be investigated as monitoring techniques would become less efficient at long distances.} Overall, while the key application of such an analysis could be in satellite-based systems, the whole area of QKD security under unconventional assumptions is a less explored territory, which deserves more attention. { One generic direction of travel is to consider the classical limitations that one can impose on Eve. This work was effectively concerned with limiting the size of an eavesdropping object, but this can be extended to other classically measurable attributes of Eve.} We hope that works like this manuscript can open new avenues of research in this area.

\section*{Acknowledgments}
M.R. is grateful to Norbert L{\"u}tkenhaus, Xiongfeng Ma, and Charles C. W. Lim for fruitful discussions around the security analysis. This work has been partially sponsored by the White Rose Research Studentship, the EPSRC via the UK Quantum Communications Hub with Grant Nos. EP/M013472/1 and EP/T001011/1, and the European Union's Horizon 2020 research and innovation programme under the Marie Sklodowska-Curie grant agreement number 675662 (QCALL). 
M.G. would like to additionally acknowledge support from the European Union via \textquotedblleft Continuous Variable Quantum Communications\textquotedblright\ (CiViQ, Grant agreement No. 820466). F.G. and H.K. acknowledge support from the Deutsche Forschungsgemeinschaft (DFG, German Research Foundation) under Germany’s Excellence Strategy - Cluster of Excellence Matter and Light for Quantum Computing (ML4Q) EXC 2004/1 -390534769. H.K. also acknowledges support by the QuantERA project QuICHE, via the German Ministry for Education and Research (BMBF Grant No. 16KIS1119K). 

All data generated in this paper can be reproduced by the provided methodology and equations.

\section{Appendices}
\appendix

\section{Estimating $\eta_{\rm AE}$ and $\eta_{\rm EB}$ parameters}
\label{Carlo}

In this Appendix, we find nominal values for parameters $\eta_{\rm AE}$ and $\eta_{\rm EB}$ if Alice and Bob are equipped with the LIDAR technology for detecting unwanted objects around them. 

\subsection{Optical Setup}
\label{sec:setup}

In this section we specify the optical setup considered in our calculation for the two authorized QKD parties, Alice (A) and Bob (B), and the eavesdropper, Eve (E). We assume A is located on a low earth orbit (LEO) satellite, travelling in a circular orbit at an altitude $L$ above the ground. It is equipped with a QKD source and a telescope with aperture radius $r_A$. B is instead placed on the surface of the Earth and he collects the light sent by A using a telescope with radius $r_B$. We address the static situation in which the satellite is at a fixed position right above the optical ground station, so that the length of the link is exactly $L$. In the following calculations, we will allow E to have two distinct satellites, one for collecting and one for re-sending the light, with appropriate values of the aperture radius and position. However, it turns out that the configuration of a single satellite is indeed optimal for her. We can therefore assume that E is represented by a spacecraft equipped with two telescopes, one for collection (pointed towards A) and one for transmission (pointed towards B), both of radius $r_E$. We also assume, as the worst-case scenario, that the aperture of the telescope represents the whole projected area of E's spacecraft.  

We assume that A's telescope sends the QKD signals in the form of a Gaussian beam, with initial beam width $W_0$, equal to its radius $r_A$, at wavelength $\lambda$. For the light propagation we neglect the action of the atmosphere and the contribution of pointing errors. We use the standard expressions for Gaussian optics, corrected through the quality factor $M^2$ in order to replicate the far-field divergence of real optical elements. E's telescope is instead perfect, meaning that she can send Gaussian beams with $M^2=1$. 

In the following we will call $z$ the coordinate along the propagation path, so that A is at $z=0$ and B at $z=L$. After a propagation of length $z \in [0,L]$, the beam width can be expressed as
\begin{align}
\label{Wz}
W(z)=W_0 \sqrt{1+\bigg(\frac{z M^2}{z_R}\bigg)^2} \ ,
\end{align}
where $z_R=\pi W_0^2/\lambda$ represents the Rayleigh range of the beam. Comparison between the far-field divergence of a perfect Gaussian beam and the divergence measured for the Micius satellite suggests a value $M^2 \approx 3$. The transmittance of such a beam, when impinging at the centre of a circular collecting aperture of radius $\rho$ can be expressed as
\begin{align}
\label{eta}
\eta(\rho,z)=1-\exp\bigg[-2 \frac{\rho^2}{W^2(z)}\bigg].
\end{align}
This expression can be used to compute the transmittance of A's beam through B's telescope, by setting $z=L$ and $\rho=r_B$:
\begin{align}
\label{etaAB}
\eta_{\rm AB}=1-\exp\bigg[-2 \frac{r_B^2}{W^2(L)}\bigg] \ ,
\end{align}
which describes the efficiency of the QKD channel, apart from additional losses like atmospheric absorption, detection efficiency and transmittance of the optical elements. 
The same formula can express the efficiency with which E can collect A's signals, while she is at position $z$ and has a collecting aperture of radius $r_E(z)$
\begin{align}
\label{etaAE}
\eta_{\rm AE}(z)=1-\exp\bigg[-2 \frac{r_E(z)^2}{W^2(z)}\bigg] \ .
\end{align}
We assume here that E is positioned at the exact centre of the beam.
The way we model the dependence of $r_E(z)$ on the distance from A and B will be specified in the next section.

We can use a similar approach to estimate the ability of E to re-send the signals that she has intercepted towards B. In order to take full advantage of her optical system, we allow E to send focused beams. It is not necessary to take this into account in the case of A, because for a typical LEO satellite the total propagation length $L$ is much larger than the Rayleigh range $z_R \approx 70$~km, so focusing would not give much advantage. For our calculations, we suppose that E has a lens of focal length $f$ just in front of her sending aperture. We can then use the ray transfer matrix formalism and obtain the following expression for the optimized width of a focused beam at distance $d$ from its transmitter \cite{SalehTeichBook}:
\begin{align}
\label{WE}
W_E(z)=\frac{\lambda d}{\pi r_E(z)} = \frac{\lambda (L-z)}{\pi r_E(z)} \ ,
\end{align}
which agrees with \cref{Wz} when $z \gg z_R$ and $r_E=W_0$. Now, using \cref{eta}, we can compute the transmittance of E's beam through B's aperture as follows
\begin{align}
\label{etaEB}
\eta_{\rm EB}(z)=1-\exp\bigg[ -2 \frac{r_B^2}{W_E(z)^2} \bigg] \ .
\end{align} 
We point out that, even in this case, the dependence of $r_E(z)$ on the length of A-E and E-B links is important and will be modelled in the next section.

\subsection{Techniques for Channel Monitoring}
\label{monitor}

In this section we obtain an upper bound on the size of an undetected E's spacecraft, depending on the distance from A's or B's position, if some sort of channel monitoring system is employed. 
Typical techniques are RADAR, LIDAR and direct optical detection. We will not analyze the last one, as it requires rather stringent conditions: E's spacecraft must be illuminated by the sun while the receiver is in eclipse and the sky must be clear. A RADAR is very power-consuming, so we will address this technique as operated only from B, on the ground (although examples of radars on spacecrafts can also be found). LIDARs instead require much less power and share similar optical elements as those used for QKD, so may be placed on both A's and/or B's sides.

The operation of a RADAR/LIDAR system can be described by the so-called RADAR equation:
\begin{align}
d_{\rm max}=\bigg(\frac{P_T G^2 \lambda \sigma}{P_{\rm min} (4\pi)^3 \kappa}\bigg)^{1/4} \ ,
\end{align}
which expresses the maximum distance at which an object with radar cross section $\sigma$ can be detected.
We are interested in the inverse dependence for the maximum $\sigma(z)$, for a space object at location $z$, i.e. distance $L-z$ from B, which is given by
\begin{align}
\label{radar}
\sigma(z)= & \frac{P_{\rm min}(4 \pi)^3 \kappa d_{\rm max}^4}{P_T G^2 \lambda^2} \notag \\
= & \frac{P_{\rm min}(4 \pi)^3 \kappa (L-z)^4}{P_T G^2 \lambda^2}.
\end{align}
Here, $P_{min}$ represents the minimum power measurable by the receiving system, $P_T$ is the total power emitted, $G$ is the gain of the radar antenna, and $\kappa$ is a parameter that accounts for all additional sources of loss.

In order to assess the applicability of a RADAR system on Bob's end, we use the following parameter values:
\begin{itemize}
\item $G=\frac{4\pi \ E \ \pi \ r_{\rm ant}^2}{\lambda_{R}^2}$ where $E=0.6$ is the antenna efficiency, $r_{\rm ant}=2$~m is the radius of the circular parabolic antenna and $\lambda_{R}=4$ cm is the wavelength of the radar signals. We chose $r_{\rm ant}=2$~m as a reasonable size for a dish to be put alongside an optical ground station.
\item $P_T=10^5$ W, as it is the power usually used in systems of this size (like the ones used in airports).
\item $P_{\rm min}=k_B T F_n B$, with $k_B$ the Boltzmann constant, $T$ the temperature, $F_n=8$ dB is the so-called noise figure and $B=2.5 \times 10^6$ Hz is the effective noise bandwidth of the setup.
\item $\kappa=7$ dB takes into account attenuation from atmospheric effects, filters and other sources.
\item We also assume that $L=500$~km corresponding to a LEO satellite. 
\end{itemize}  
In general the radar cross section $\sigma$ is not equal to the geometric projected area and it strongly depends on the shape of the object. Only for spherical objects these two quantities coincide and this is the case we consider here. In this way, we can set the radius of E's telescope to $r_E=\sqrt{\sigma/\pi}$. Figure~\ref{fig:radar} shows the minimum size of $r_E$, calculated from Eq.~(\ref{radar}) at the above parameter values, if a radar is located at Bob's site, i.e., at $z=L$. Figure~\ref{fig:radar} suggests that, if we only use RADAR at Bob's end, we can easily miss eavesdropping objects of a few meters in radius. This implies that we may not achieve useful bounds on $\eta_{\rm AE}$ and $\eta_{\rm EB}$, in Eq.~(\ref{etaAE}) and Eq.~(\ref{etaEB}), if we only rely on RADAR as a monitoring system. Even assuming that low-power radar could be employed on the satellite to monitor the first tens of km around it, a telescope of 3~m in radius at 100~km from A would be able to intercept and resend with transmittances very close to 1. In practice, radar techniques are currently used to monitor the number of objects present in low orbits around the Earth \cite{Klinkrad2004}. However, much bigger facilities (antenna radius $\gtrsim 10$ m) are necessary for such missions and the information is usually not in real-time, but used to build and update catalogues of the objects. {We would therefore consider the radar solution insufficient for our purposes, while passive monitoring could always provide additional information}. We next consider the LIDAR option.

\begin{figure}[b]
\includegraphics[width=0.5\textwidth-15pt]{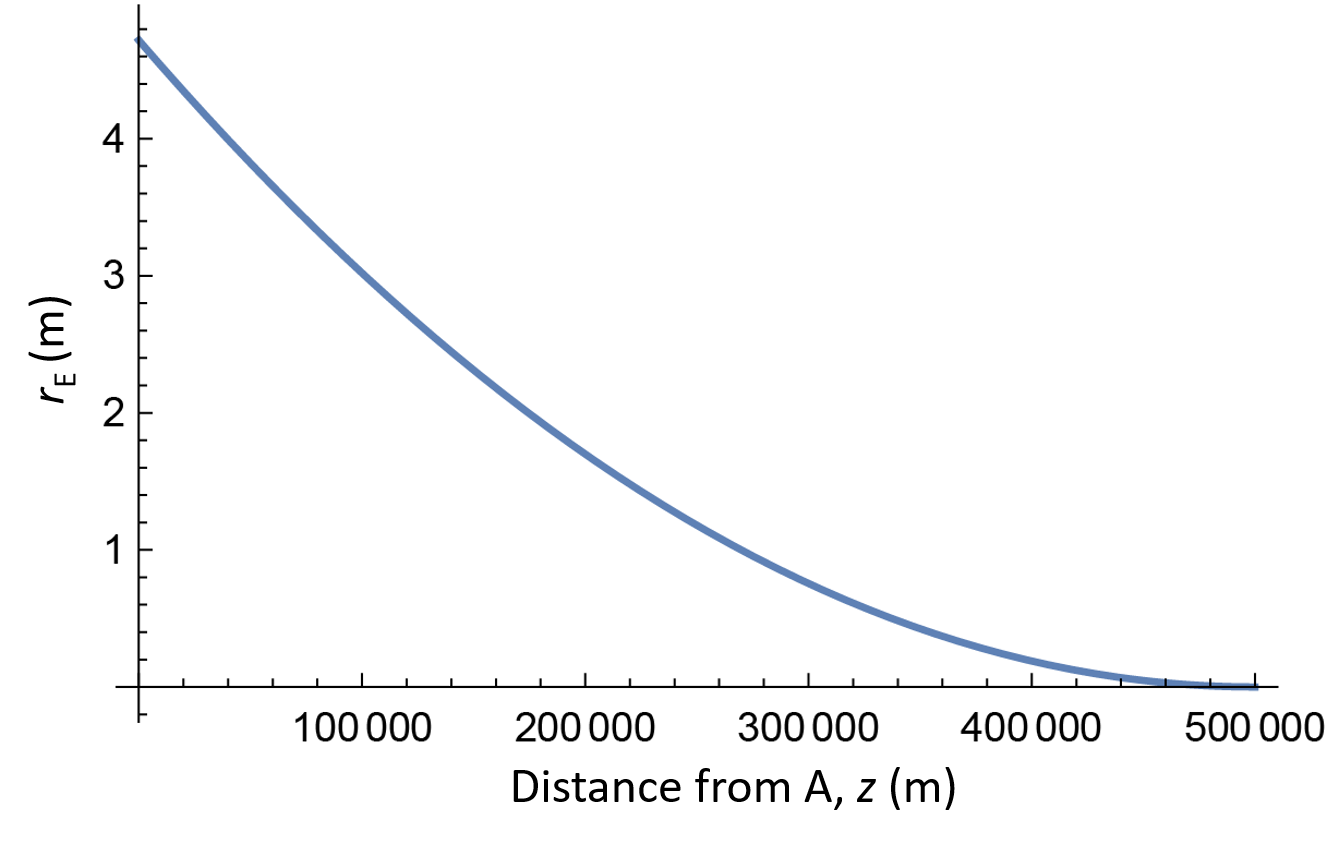}
\caption{The minimum radius for E's telescope aperture detectable by a typical radar system located at Bob's site. Note that the distance from Bob is measured by $L-z$. An object 500~km away from Bob must have a radius greater than 4~m, for our chosen parameter values, to be detectable by Bob's radar system.}
\label{fig:radar}
\end{figure}

Much better performance can be achieved using LIDARs. The working principle is the same as radars, but in this case light in the near ultraviolet, visible, or near infrared range is sent and recorded after reflection from the object under study. In this case, instead of enormous antennas, we only need telescopes of reasonable sizes. For example, the same telescopes used for exchanging QKD signals, or alignment, can be used for LIDAR operation. Moreover, instead of powers of tens of kW, lasers with power on the order of 1~W are sufficient, meaning that this technique can rather easily be implemented on even small satellites, as well as on Bob's side. As expected, the big advantage comes from the much shorter wavelength of the employed light with respect to the microwave signals used in the radar technique, resulting in much smaller diffraction of the electromagnetic beams.

In this case, we can try to use again the standard radar equation of Eq.~(\ref{radar}), with suitably chosen parameters. We report here a simple calculation, using again Gaussian optics, that gives a result very similar to the radar equation (with LIDAR parameters), for when the LIDAR is placed on the satellite. A similar calculations can be used for a LIDAR based in the ground station. 
We use Eq.~(\ref{WE}) and modify it to take into account the realistic quality factor $M^2$ as estimated before
\begin{align}
W_{\rm LIDAR}(z)=\frac{\lambda_{\rm LIDAR} z M^2}{\pi W_0} \ ,
\end{align}
where $\lambda_{\rm LIDAR}$ is the LIDAR wavelength.
The intensity distribution of such a beam can be expressed as
\begin{align}
\label{intensity}
I(r,z)=\frac{2 P_T}{\pi W_{\rm LIDAR}(z)^2} \exp \bigg[ \frac{2 r^2}{W_{\rm LIDAR}(z)^2} \bigg] \ ,
\end{align}
where $P_T$ is the total power carried by the beam and $r$ is the distance from the beam centre in the plane transversal to the direction of propagation. We assume that the reflecting object is at the centre of the beam.

We compute the total power incident on the object integrating Eq.~(\ref{intensity}) in the area corresponding to E's spacecraft as follows
\begin{eqnarray}
P(z)&=&\int_{|r|<r_E} I(r,z) dr d\theta= \nonumber \\
&=&P_T \bigg(1-\exp\bigg[-\frac{2 r_E^2}{W_{\rm LIDAR}(z)}\bigg]\bigg) .
\end{eqnarray}
We assume that the light is reflected back isotropically by the object under study, with reflectivity $\alpha$, resulting in a received light intensity of
\begin{eqnarray}
I_R(z)&=&\frac{P(z) \alpha}{4 \pi z^2}= \nonumber \\
&=&\frac{P_T \alpha}{4 \pi z^2} \bigg(1-\exp\bigg[-\frac{2 r_E^2}{W_{\rm LIDAR}(z)}\bigg]\bigg) \ .
\end{eqnarray}
The total collected power reaching the satellite LIDAR is then $P_R(z) = I_R(z) \pi r_A^2 \kappa$, where we account for any additional loss encountered during transmission and collection by factor $\kappa$. In order to obtain the bound on the size of E's object, we can then invert this expression and equate $P_R(z)$ to the minimum power $P_{\rm min}$ measurable by the receiving setup, as follows:
\begin{align}
\label{A}
r_E(z)^2=-\bigg(\ln\bigg[1-\frac{2 P_{\rm min} k z^2}{\alpha P_T W_0^2}\bigg]\bigg)\bigg(\frac{\lambda_{\rm LIDAR} z M^2}{\pi W_0}\bigg)^2 \ .
\end{align}

For the rest of this section, unless otherwise noted, we use the following parameter values. We set $\lambda_L=800$ nm and assume $\kappa=0.25$. The transmitted power is set to $P_T=1$~W due to the limit on the power consumption on small satellites. For the ground-based LIDAR, this value could even be higher, although, offering a small advantage, as we show by the end of this section.
We choose a rather conservative value for the reflectivity of the object, $\alpha=0.1$, considering that for different metals it is usually around $\alpha=0.5$ or more. Coating can be used to lower this value, however, measurements at different wavelengths could limit the effectiveness of this technique. We also assume that $M^2 = 3$, $r_A=15$~cm, and $r_B=50$~cm. These values are compatible with the instruments used in the Chinese satellite mission Micius. All other relevant parameters are the same as the radar case. 

Figure~\ref{fig:lidar1} shows the estimated maximum radius of E's object that does not trigger our LIDAR monitoring system, versus its distance from the satellite. The results obtained by using Eq.~(\ref{A}) and Eq.~(\ref{radar}) are both shown. They differ because the efficiency of the transmitter and the reflectivity of the object are modelled in different ways. We see that the bound on the size of undetectable objects, $r_E$, is much smaller as compared to the values shown in \cref{fig:radar} using the radar technique, giving hope that the values obtained for $\eta_{\rm AE}$ and $\eta_{\rm EB}$ in this case may be low enough to be useful in the enhancement of the secret key rate. 

\begin{figure}[b]
\includegraphics[width=0.5\textwidth-15pt]{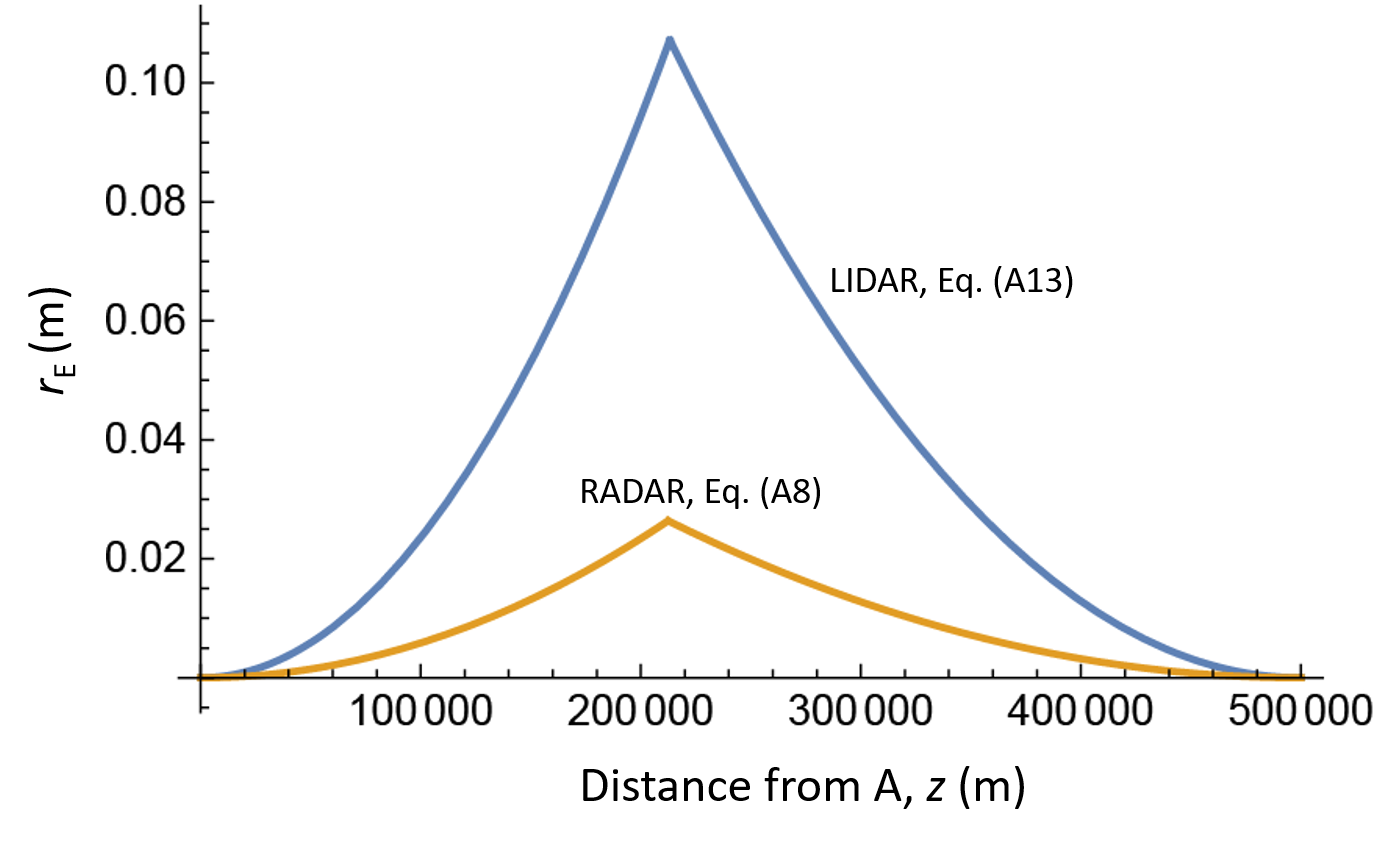}
\caption{Minimum radius for E's telescope aperture detectable by LIDAR measurements performed, simultaneously, from the satellite and from the ground. The bound on this quantity obtained from two different techniques: The blue curve is obtained from Eq.~(\ref{A}), while the orange curve from the radar equation Eq.~(\ref{radar}) using parameters suitable for a LIDAR system.}
\label{fig:lidar1}
\end{figure}

The minimum measurable power $P_{\rm min}$ used in Fig.~\ref{fig:lidar1} is obtained by calculating the background light collected by the satellite in normal working conditions. For the LIDAR placed on the satellite, the main source of background light during night-time operation is represented by the light of the Moon reflected by the Earth \cite{Bonato_2009}, which can be expressed as follows
\begin{align}
P_{\rm min}^A=\alpha_E \alpha_M R_M^2 r_A^2\frac{\Omega_{fov}}{d_{EM}^2}H_{sun} B_{filter}\ ,
\end{align}
where $\alpha_E$ and $\alpha_M$ are the albedo of Earth and Moon, $R_M$ is the radius of the Moon, $d_{EM}$ is the Earth-Moon distance, $H_{sun}$ is the Sun irradiance at $\lambda_L$ and $\Omega_{fov}$ is the field of view of the telescope and $B_{filter}$ is the bandwidth of the spectral filters.
For the LIDAR on the ground, we estimate the background light from the analysis in \cite{miao}, as follows
\begin{align}
P_{\rm min}^B=H_b \Omega_{fov} \pi r_B^2 B_{filter} \ ,
\end{align} 
where $H_b$ is the brightness of the sky background. The typical value for such background lights is very small suggesting that in order to obtain some statistics about such sources we may need to use single-photon detectors in our LIDAR system \cite{tachella2019}.

The previous analysis does not take into account the fact that the LIDAR detection from the ground will be strongly affected by the presence of the atmosphere. The air can back-scatter the light sent by Bob's LIDAR, especially when the sky is not completely clear, giving a signal that can be attributed to Eve's object. This means that, without additional analysis, every time we will measure a reflected power greater than $P_{\rm min}$, we will think that this is because of Eve's apparatus and the measured power will be used to bound its size. If part of the back-scattered light is due to the atmosphere, we will end up over-estimating the size of Eve's object, and consequently its collecting efficiency. In that sense, while this issue can loosen our lower bound on the key rate, it does not make our analysis unreliable.

\subsection{Bounds on $\eta_{\rm AE}$ and $\eta_{\rm EB}$}
\label{etas}

In this section we report the numerical results for the E's collecting and re-sending efficiencies, obtained using the analysis provided in the previous sections. Figure~\ref{fig:eta} shows the values of $\eta_{\rm AE}$ and $\eta_{\rm EB}$, computed, resepctively, from Eq.~(\ref{etaAE}) and Eq.~(\ref{etaEB}), as a function of $z$. In both graphs, the maximum value happens somewhere in the middle of the orbit. This is because we are using LIDAR on both A and B, and the maximum value is achieved at the point where E's telescope is the biggest, which is roughly in the middle. This happens because the width of the beams, during the propagation, vary linearly with $z$, while the bound on E's size is proportional to $z^2$ (equivalently, the cross-section in Eq.~(\ref{radar}) is proportional to $z^4$). We see that $\eta_{\rm AE}$ remains below 0.1, while $\eta_{\rm EB}$ grows up to about 1. There are two main reasons for this behaviour. First, we allow E to use perfect optics that generate Gaussian beams with minimal divergence and second, B's telescope aperture is bigger than A's.

\begin{figure}[t]
\includegraphics[width=0.5\textwidth-15pt]{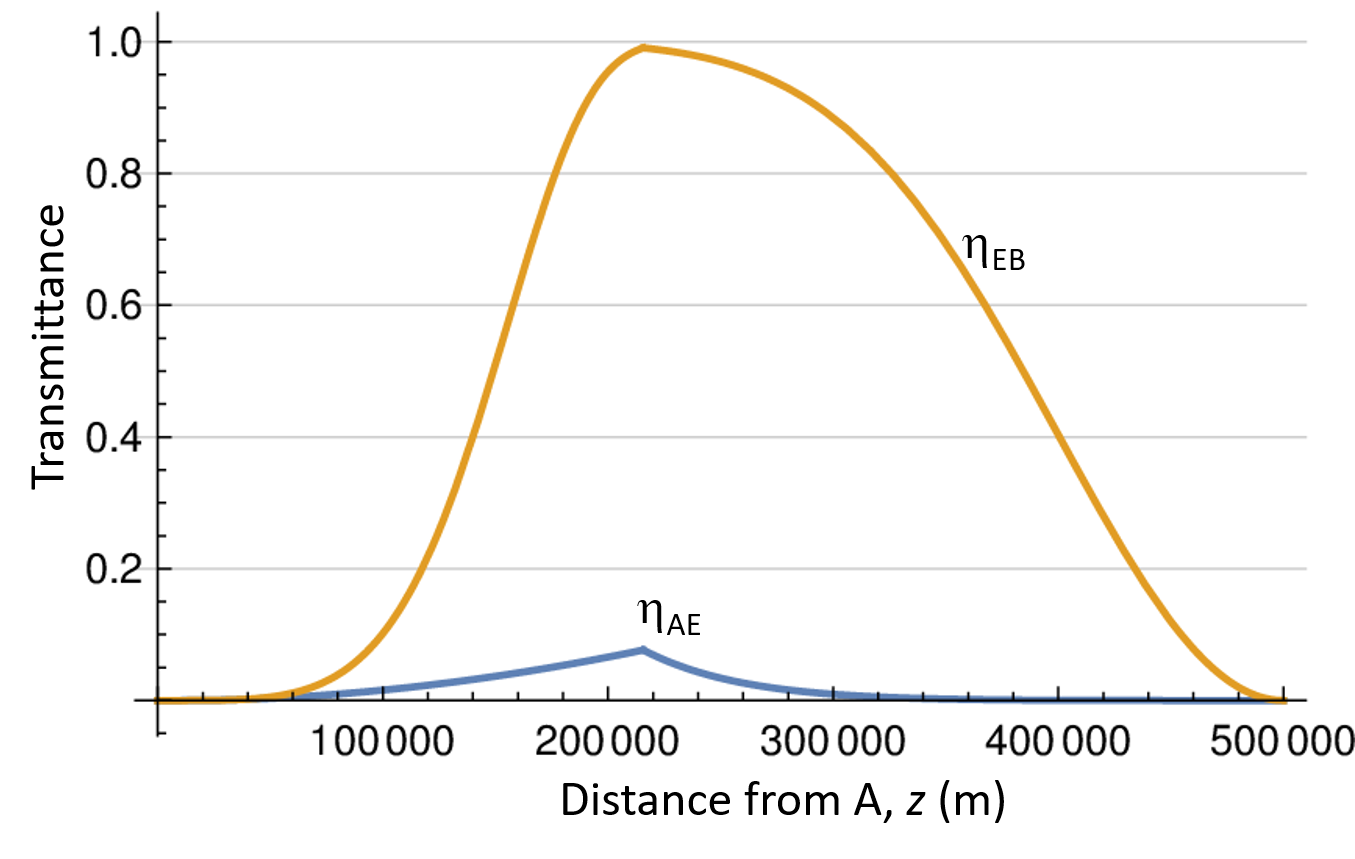}
\caption{Values of $\eta_{\rm AE}$ and $\eta_{\rm EB}$, for an undetected Eve, computed using Eq.~(\ref{etaAE}) and Eq.~(\ref{etaEB}), respectively.}
\label{fig:eta}
\end{figure}

Figure~\ref{fig:rs} shows the values of some quantities of the setup as a function of the coordinate $z$, useful to understand the behaviour observed in Fig.~\ref{fig:eta}. The $r_E$ curve close to the x-axis is the same as the upper curve in Fig.~\ref{fig:lidar1}, which shows the maximum radius of the undetected E. The $W_E$ curve represents the width of the beam, sent by E at distance $z$ from A with a telescope of radius $r_E(z)$, when it arrives at B's receiving plane. The $W_{\rm LIDAR}$ curve is, instead, the width of the beam sent by A as it propagates towards B. We see that when it arrives at B, after 500~km of travelling, the beam is about 2.5~m in radius, which is several times larger than that of B's telescope, giving a transmittance between the legitimate parties of $\eta_{\rm AB}=0.05$ (only considering diffraction losses, without collection and detection losses). As for Eve, however,  the minimum of the $W_E$ curve is roughly 30~cm at B, which is smaller than B's telescope size, resulting in $\eta_{\rm EB}\simeq 1$. { Note that $W_E$, in \cref{WE}, is inversely proportional to $r_E(z)(L-z)$, which justifies the asymmetry in the graph.}

\begin{figure}[t]
\includegraphics[width=0.5\textwidth-15pt]{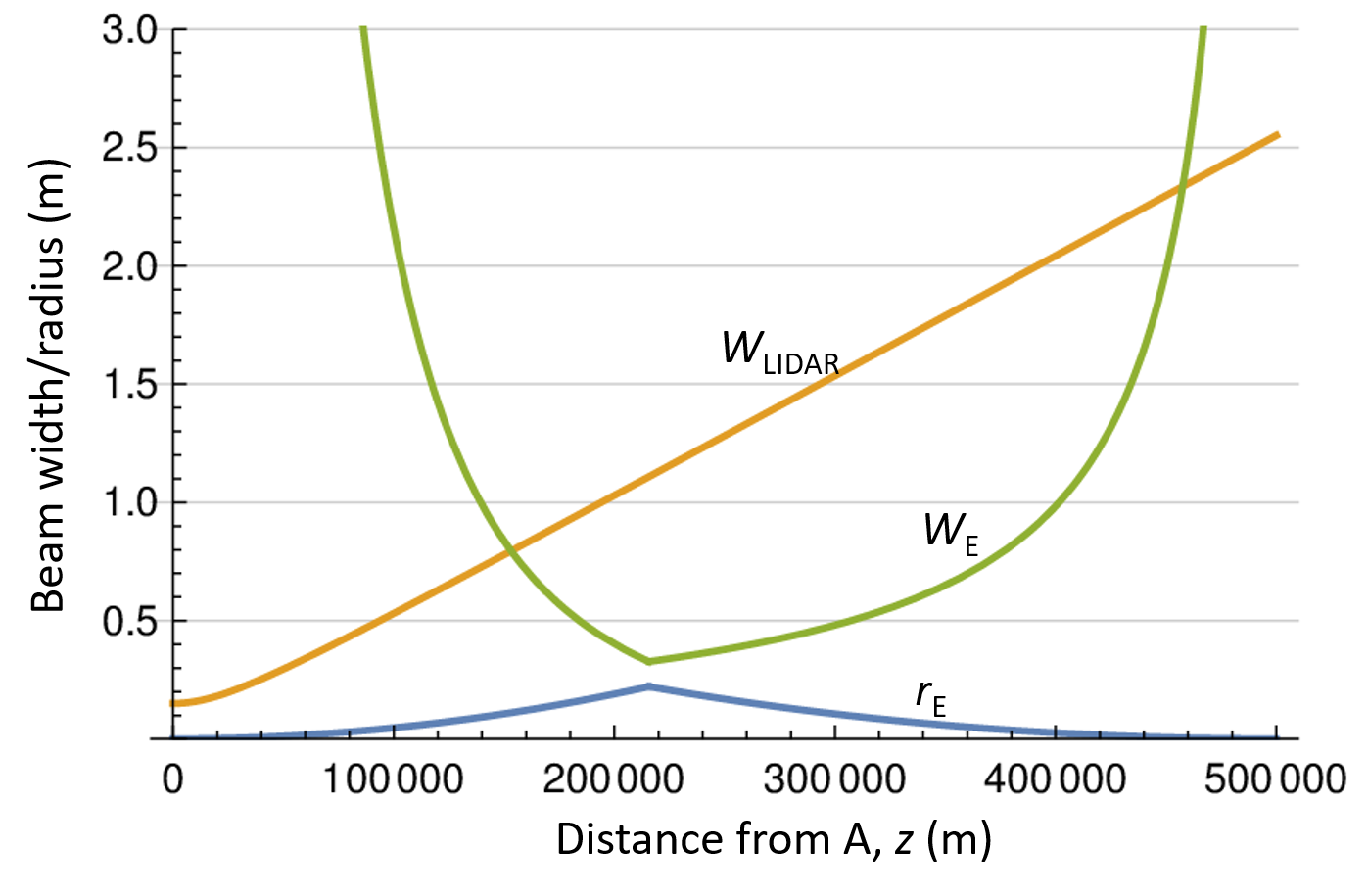}
\caption{The maximum radius, $r_E$, of the undetected Eve's object, the width of the propagating beam, $W_{\rm LIDAR}$, sent by Alice, and the width of signal sent by Eve at point B, $W_E$, versus $z$.}
\label{fig:rs}
\end{figure}

The values in Fig.~\ref{fig:eta} can be lowered by raising the value of LIDAR's transmitted power. Notice that $r_E \propto P_T^{1/2}$, so if we raise the power by a factor 4, to 4 W, the bound on Eve size will be halved. In this case, smaller values of $\eta_{\rm AE}$ and $\eta_{\rm EB}$ are expected, as shown in Fig.~\ref{fig:eta2}. $\eta_{\rm AE}$, in particular, reaches a maximum of about 3\%, giving big room for improvement in the achievable key rate. This bound very strongly depends on the minimum measurable power $P_{\rm min}$. Any improvement in the filtering techniques (defined by the parameters $B_{filter}$ and $\Omega_{fov}$) will improve the performance. In the same way, going to lower wavelengths will reduce the diffraction losses and improve the bound. We point out that, in practice, the monitoring can possibly be repeated with a rather low frequency, leaving the remaining time for the QKD signal exchange. This means that the power actually consumed during monitoring operation should be manageable even by small satellites. {On the other hand, if QKD missions are merged with remote sensing missions used for earth observations, then large satellite payloads, and therefore, high-power LIDAR systems can be used, which considerably improve the bounds on $\eta_{\rm AE}$ by one to two orders of magnitude. Examples include 562~W LIDAR used in CALIPSO and 1865~W in LITE missions.}

\begin{figure}[ht!]
\includegraphics[width=0.5\textwidth-15pt]{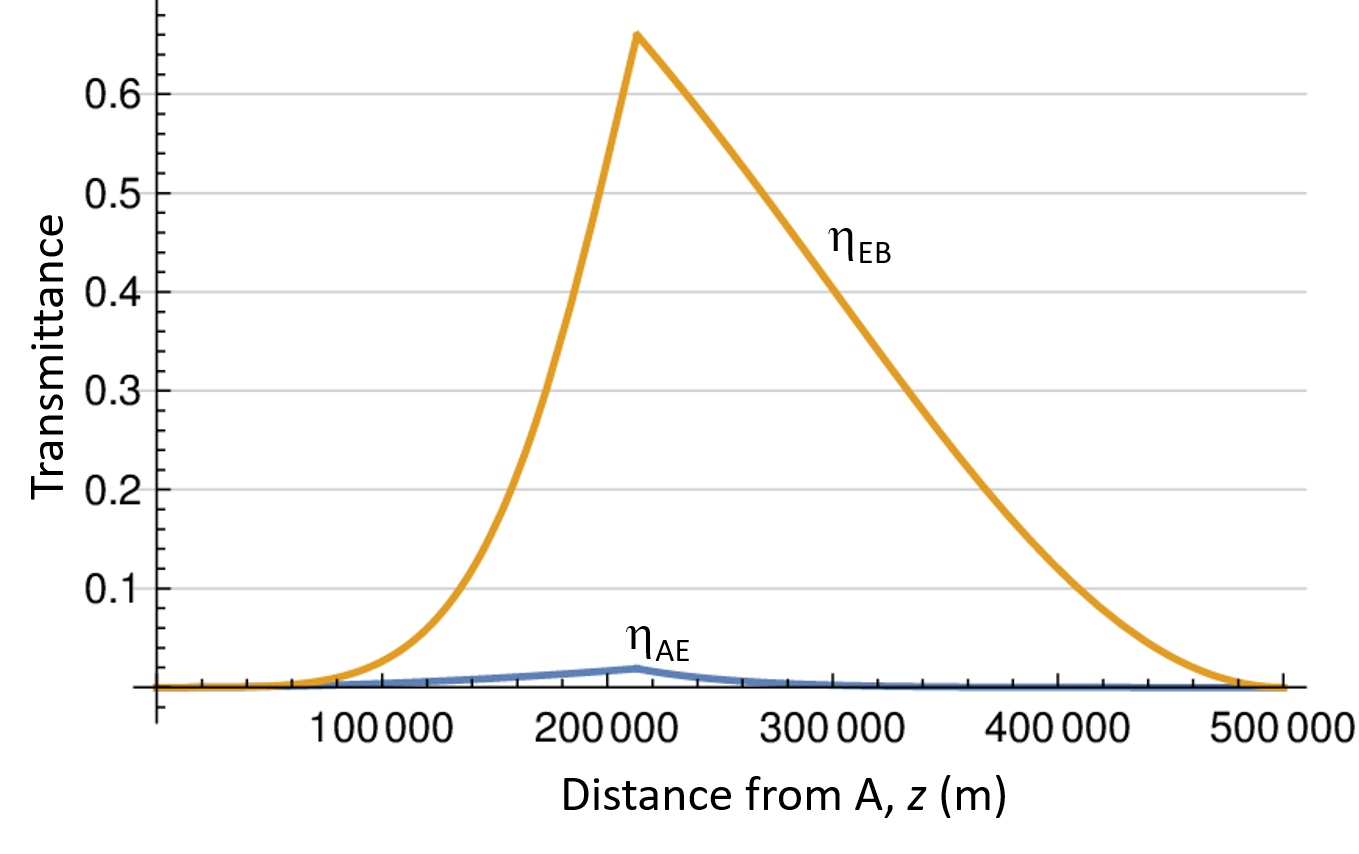}
\caption{Values of $\eta_{\rm AE}$ and $\eta_{\rm EB}$, for an undetected Eve, computed using Eq.~(\ref{etaAE}) and Eq.~(\ref{etaEB}), for a power of 4 W.}
\label{fig:eta2}
\end{figure}

The LIDAR technique, in the simplified approach we used in these calculations, is sensitive to the total power reflected by objects illuminated by the transmitted light. This means that we are safe even in the situation where Eve places more flying objects, which taken alone would be smaller than the detectable size. If we detect that an object or more are passing between A and B, by measuring a received power $P_R>P_{\rm min}$, we can assume that they are all malicious, estimate their size by replacing $P_{\rm min}$ with $P_R$ in the expressions above and bound $\eta_{\rm AE}$ and $\eta_{\rm EB}$ in the real case. 

We point out again that the presence of back-reflections from the atmosphere would give an over-estimation of the size of Eve when measured from B, which has not been considered here, leading to higher values of $\eta_{\rm AE}$ and $\eta_{\rm EB}$. More sophisticated techniques should be able to address this problem, for instance, using the timing information obtained when using the LIDAR in the pulsed regime. The advantage introduced by sending a beam with higher power, analyzed in Fig.~\ref{fig:eta2}, would be less effective for B, because it would correspond to more light back-reflected by the atmosphere, too.

Until now we have considered the static case where the satellite is fixed at the position closest to the ground station. 
We study now how the maximum values of $\eta_{\rm AE}$ and $\eta_{\rm EB}$ (optimal for E) vary during the passage of the satellite. We show the results in Fig.~\ref{fig:vsL} at $P_T=$1 W transmitted power for the LIDAR system, and in Fig.~\ref{fig:vsL2} at $P_T=4$~W. As can be seen, both configurations perform well for high elevation angles, however, the higher power level is required to put useful bounds at low elevation angles. As pointed out before, if the available power output is limited, one can achieve the same performance by changing other parameters of the setup.

\begin{figure}[ht!]
\includegraphics[width=0.5\textwidth-15pt]{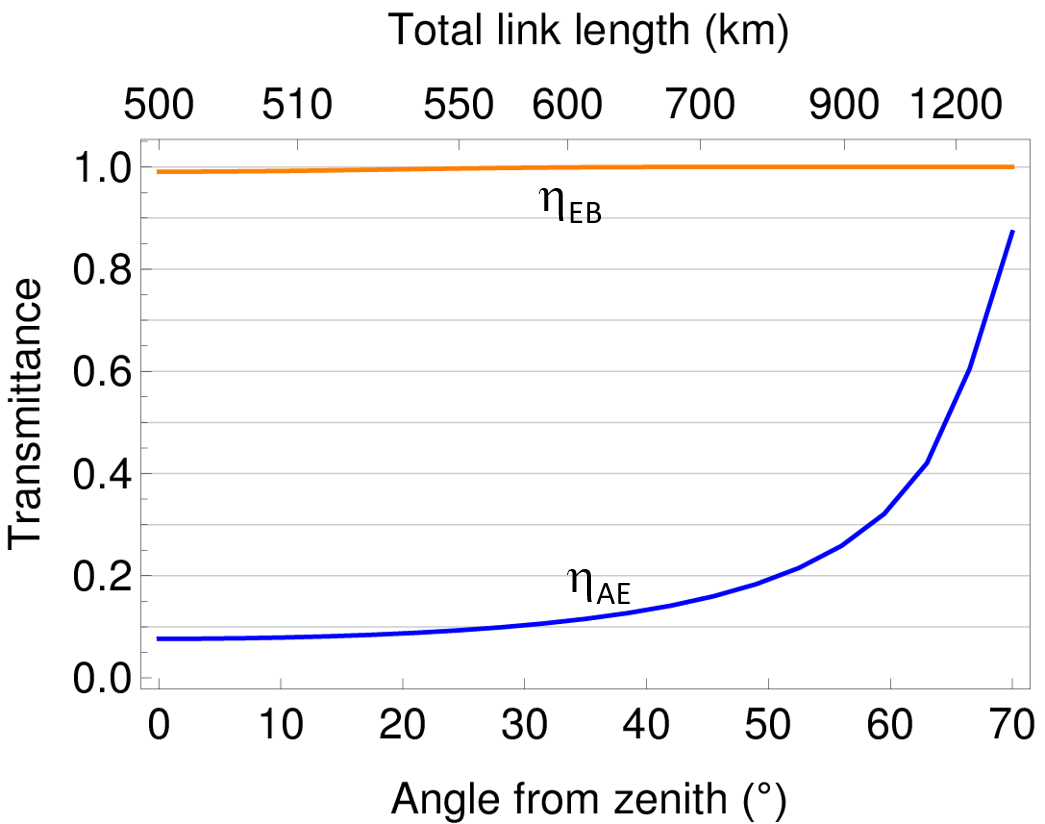}
\caption{Maximum values of $\eta_{\rm AE}$ and $\eta_{\rm EB}$, for an undetected Eve, as a function of the position of the satellite, for a LIDAR transmitted power of 1~W.}
\label{fig:vsL}
\end{figure}

\begin{figure}[ht!]
\includegraphics[width=0.5\textwidth-15pt]{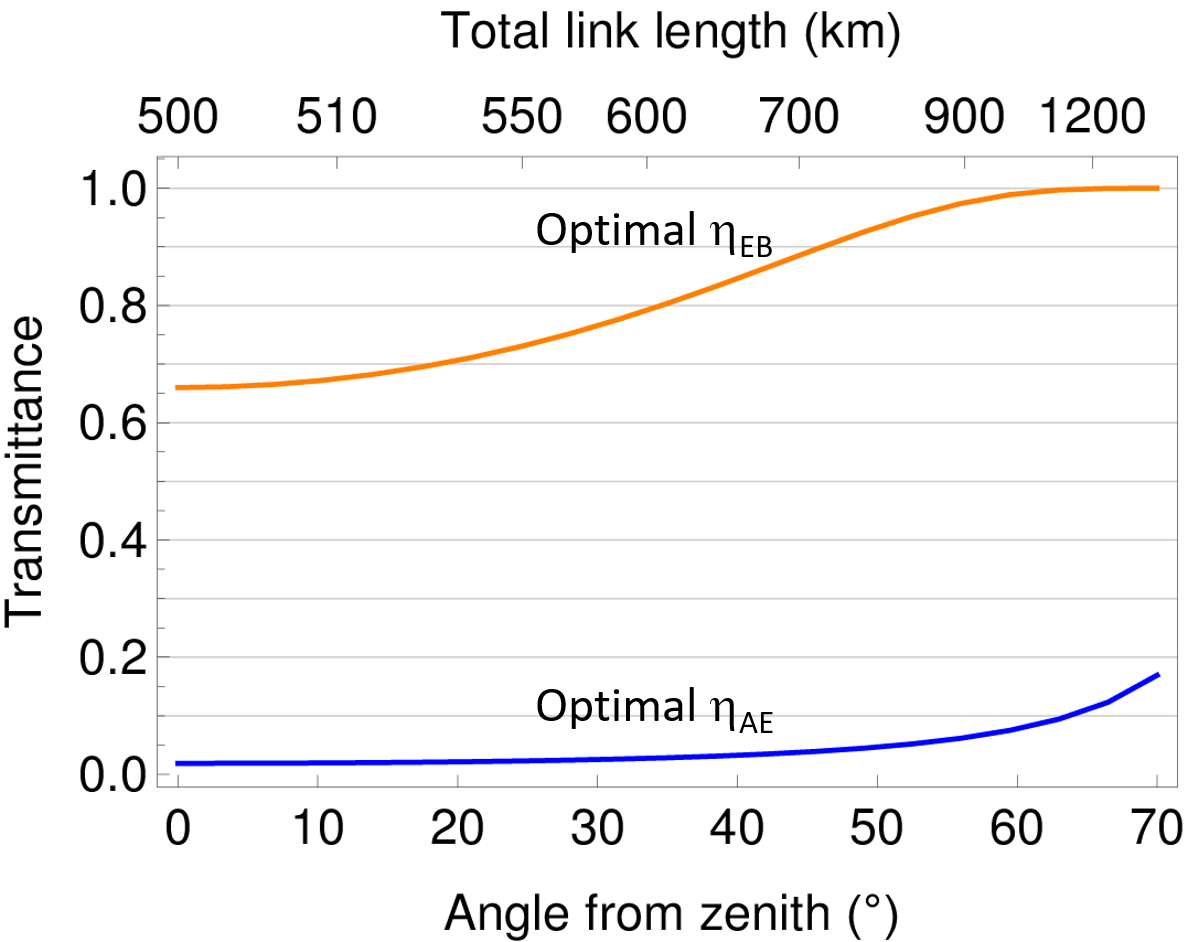}
\caption{Maximum values of $\eta_{\rm AE}$ and $\eta_{\rm EB}$, for an undetected Eve, as a function of the position of the satellite, for a LIDAR transmitted power of 4~W.}
\label{fig:vsL2}
\end{figure}

For comparison, we report in Fig.~\ref{fig:etaAB} the behaviour of $\eta_{\rm AB}$, from Eq.~(\ref{etaAB}), as a function of the position of the satellite. The upper curve represents only the diffraction losses, while in the lower curve other sources of loss are also considered. In particular, 50\% for detection loss, 80\% for the transmittance of the receiving optics, and absorption in the atmosphere is accounted for by $\chi_{ext}=\exp \big[ - \beta \sec(\theta) \big]$, where $\beta=0.7$ at $\lambda=800$~nm with $\theta$ being the angle from zenith. {Note that the expression for $\chi_{ext}$ is an approximate value at large values of $\theta$. We have, however, compared our results with that obtained from software tools such as MODTRAN 5, and the results are within an acceptable range for the purpose of this study.} The inclusion of pointing errors should have a fairly small impact, about 2-3 dB.

\begin{figure}[ht!]
\includegraphics[width=0.5\textwidth-15pt]{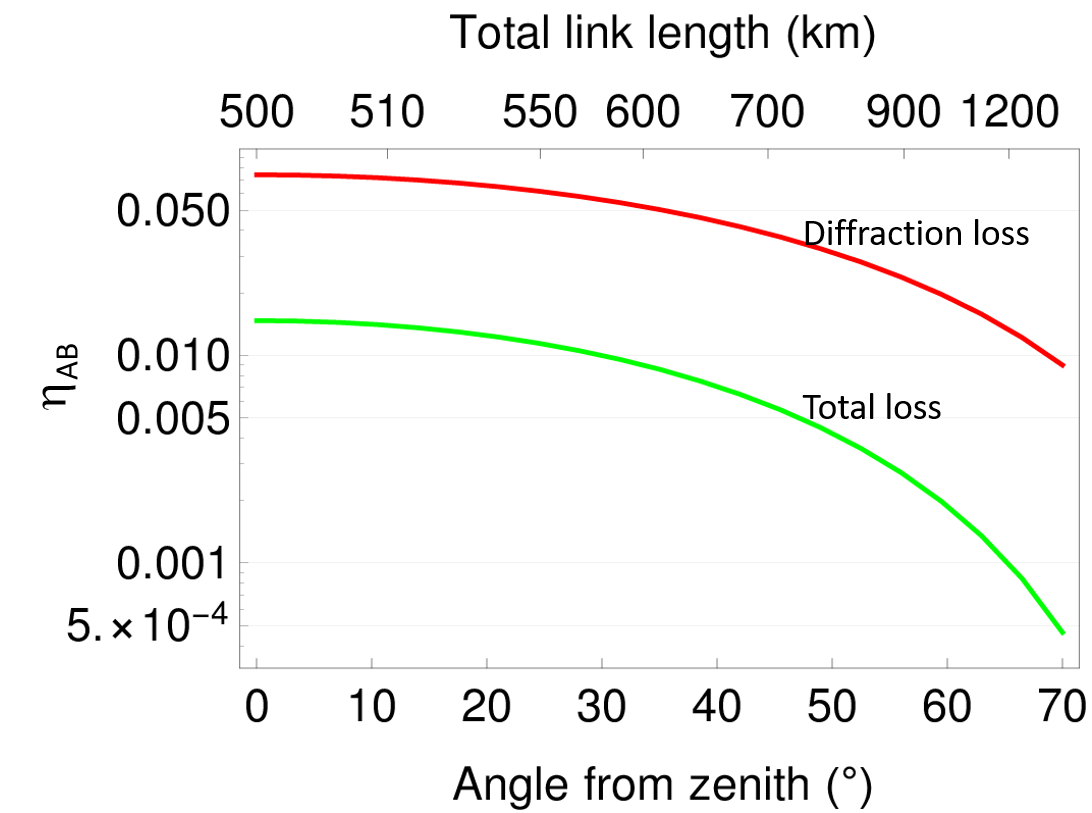}
\caption{Transmittance of the beam sent by A through B's aperture $\eta_{\rm AB}$, Eq.~(\ref{etaAB}), as a function of the position of the satellite.}
\label{fig:etaAB}
\end{figure}

In the previous analysis, we fixed the reflectivity of Eve's spacecraft to bound its size. The value chosen at the end of Sec.~\ref{monitor}, $\alpha=0.1$, is conservative enough if one considers standard spacecrafts. However, lower values of reflectivity parameters can be reached if specific technologies are used. For example, nano-structured coatings \cite{anti} can be laid over opaque surfaces, which can enable reflectivity values $<10^{-2}$. Similar values can be obtained on transparent surfaces (such as lenses), using multi-layer interferometric coatings. In Fig.~\ref{fig:albedo} we report the minimum value of reflectivity parameter of E's surfaces to achieve $\eta_{\rm AE}<1$, for different positions of the satellite with respect to the ground station. This means that, by fixing all other parameters, any value of reflectivity $\alpha < \alpha_{\rm min}$ will lead to $\eta_{\rm AE}=1$, so only values $\alpha > \alpha_{\rm min}$ lead to useful bounds in our analysis. We see from Fig.\ref{fig:albedo} that if E uses such high-performances coatings, the LIDAR setup is no longer sensitive enough. In this case, we have to compensate for the lower reflectivity by increasing the emitted power $P_T$, increasing the directionality of the beam (smaller $\lambda_L$ and/or larger $W_0$) or decreasing the minimum measurable power $P_{\rm min}$.

\begin{figure}[ht!]
\includegraphics[width=0.5\textwidth-15pt]{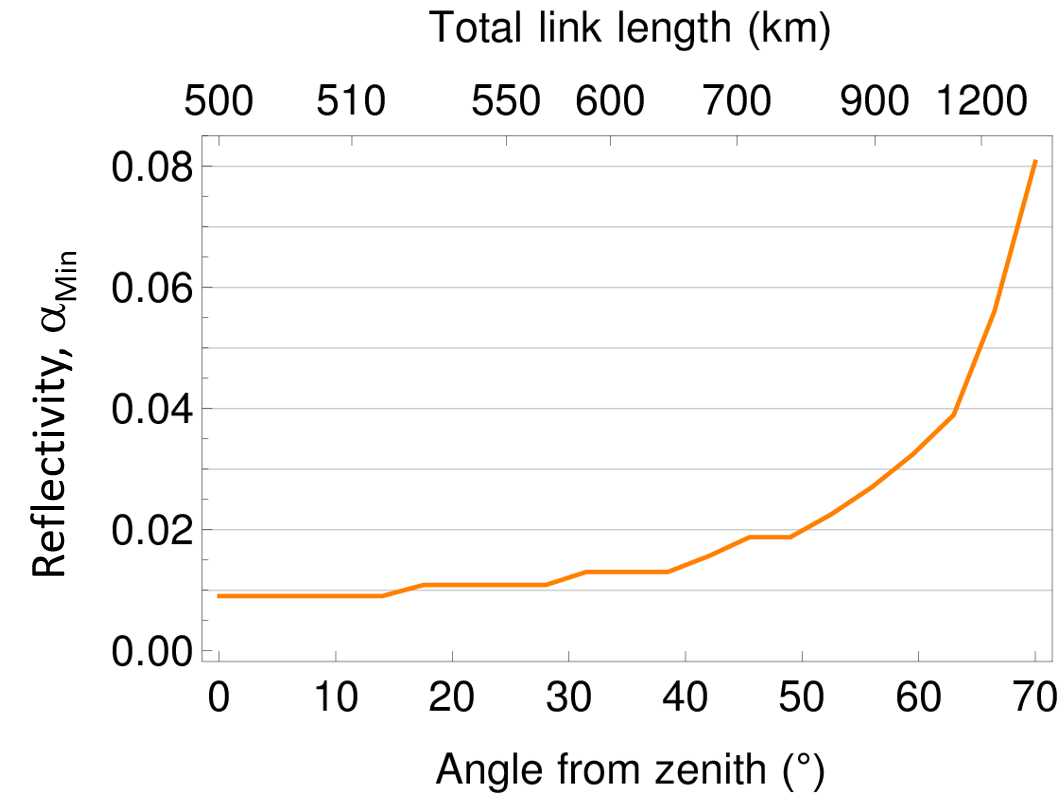}
\caption{Minimum value of reflectivity parameter of E's surfaces to achieve $\eta_{\rm AE}<1$, as a function of the angle of the satellite with respect to the zenith of the ground station.}
    \label{fig:albedo}
\end{figure}

\section{Proof of Lemma~1} \label{app:lemma1proof}

Here we prove Lemma~1.

\begin{proof}
Under the conidtion of a quantum Bob with access to $B$ and $F_0$ modes, the in-principle achievable asymptotic key rate of the QKD protocol in Fig.~\ref{fig:scenarios} is given by the Devetak-Winter bound \cite{Devetak-Winter, Pirandola:AQCrypt}:
\begin{align}
r_k=H(X|E)_k - H(X|BF_0)_k, \quad\mbox{$k=a,b$} \label{asympt-rate}
\end{align}
where, for scenario $k=a,b$ in \cref{fig:scenarios}, $H(X|E)_k$ is the conditional von Neumann entropy of Alice's classical outcome $X$ given Eve's quantum information $E$, whereas $H(X|BF_0)$ is the conditional entropy of Alice's outcome $X$ given Bob's quantum information, which includes the joint state $B F_0$ at the output of the telescope. This is effectively a classical-quantum-quantum (CQQ) scenario, where Alice has a classical state, but Eve and Bob hold on to their quantum states. 

The entropy functions in \cref{asympt-rate} are computed on the quantum states $\rho_{XE}$ and $\rho_{XBF_0}$, which in turn are the reduced density operators of the single-round global quantum state $\rho_{XB F_0 F_1 F_2 E}$. In the following, we compute the latter state for the setups in Figs.~\ref{fig:scenarios}(a) and (b). We denote with $\mathcal{B}$ the map corresponding to the beam splitters, and add subscripts on the maps' symbols to indicate the subsystems on which they act.  

The global state obtained after one round of the protocol in the setting of Fig.~\ref{fig:scenarios}(a) is given by:
\begin{widetext}
\begin{align}
\rho^{(a)}_{XBF_0 F_1 F_2 E}&=\mathcal{E}_T\circ\mathcal{B}_{BF_2}\circ\mathcal{E}'_{F_0 F_1}\circ \mathcal{E}_{BE}\circ\mathcal{B}_{BF_0}\circ M_A \nonumber \\ 
& \Big(\ketbra{\psi_{AB}}{\psi_{AB}}\otimes \ketbra{0}{0}_{F_0} \otimes \ketbra{\psi_F}{\psi_F} \otimes \ketbra{\psi_E}{\psi_E} \otimes \ketbra{0}{0}_{F_2}\Big) .
\label{global-state-b}
\end{align}
Then the reduced state on subsystems $XE$, over which the entropy $H(X|E)$ is computed, is given by:
\begin{align}
\rho^{(a)}_{XE}&=\Tr_{BF_0F_1F_2}[\rho^{(a)}_{XBF_0 F_1 F_2 E}] \nonumber\\
&=\Tr_{BF_0} \Big[\mathcal{E}_{BE}\circ\mathcal{B}_{BF_0}\circ M_A (\ketbra{\psi_{AB}}{\psi_{AB}}\otimes \ketbra{0}{0}_{F_0} \otimes \ketbra{\psi_E}{\psi_E})\Big] \label{XE-state-b},
\end{align}
where we used Kraus' theorem to remove the outer quantum maps that act on the subsystems that are traced out.

Similarly, the global state for Fig.~\ref{fig:scenarios}(b) is given by
\begin{align}
\rho^{(b)}_{XBF_0 F_1 F_2 E}&=\mathcal{E}_T\circ\mathcal{E}_{V}\circ\mathcal{B}_{BF_2}\circ\mathcal{E}'_{F_0 F_1}\circ \mathcal{E}_{BE}\circ\mathcal{B}_{BF_0}\circ M_A \nonumber \\ 
& \Big(\ketbra{\psi_{AB}}{\psi_{AB}}\otimes \ketbra{0}{0}_{F_0} \otimes \ketbra{\psi_F}{\psi_F} \otimes \ketbra{\psi_E}{\psi_E} \otimes \ketbra{0}{0}_{F_2}\Big) \label{global-state-a}.
\end{align}
Note that, compared to \cref{global-state-b}, it only presents the additional CPTP map $\mathcal{E}_{V}$. For the reduced state we obtain:
\begin{align}
\rho^{(b)}_{XE}&=\Tr_{BF_0F_1 F_2}[\rho^{(b)}_{XBF_0 F_1 F_2 E}] \nonumber\\
&=\Tr_{BF_0}\Big[\mathcal{E}_{BE}\circ\mathcal{B}_{BF_0}\circ M_A (\ketbra{\psi_{AB}}{\psi_{AB}}\otimes \ketbra{0}{0}_{F_0} \otimes \ketbra{\psi_E}{\psi_E})\Big] \label{XE-state-a}.
\end{align}
\end{widetext}

From \cref{XE-state-b} and \cref{XE-state-a} we observe that the reduced states on $XE$ are the same for both scenarios, i.e.,  $\rho^{(a)}_{XE}=\rho^{(b)}_{XE}$, which implies that:
\begin{align}
H(X|E)_a = H(X|E)_b \label{equal-entropy},
\end{align}
since the entropy functions are computed on the same quantum state.

From \cref{global-state-b} and \cref{global-state-a} we observe that $\rho^{(b)}_{XBF_0 F_1 F_2 E}$ can be obtained from $\rho^{(a)}_{XBF_0 F_1 F_2 E}$ through the following CPTP map:
\begin{align}
\rho^{(b)}_{XBF_0 F_1 F_2 E} &= \mathcal{R}_{B F_0} (\rho^{(a)}_{XBF_0 F_1 F_2 E}) \label{R-relation},
\end{align}
where 
\begin{align}
\mathcal{R}_{B F_0} := \mathcal{E}_T \circ \mathcal{E}_{V} \circ \mathcal{E}_T^{-1}\label{Rmap}.
\end{align}
By tracing over $F_1 F_2 E$ in \cref{R-relation}, the reduced state of $XBF_0$ in Fig.~\ref{fig:scenarios}(b) can be obtained by applying the CPTP map $\mathcal{R}_{B F_0}$ to the reduced state of Fig.~\ref{fig:scenarios}(a), that is:
\begin{align}
\rho^{(b)}_{XBF_0} = \mathcal{R}_{B F_0} (\rho^{(a)}_{XBF_0}). \label{reduced-R-relation}
\end{align}
By the fact that quantum maps applied on the conditioning system can only increase the conditional von Neumann entropy \cite{Nielsen-Chuang}, we have that 
\begin{align}
H(X|BF_0)_a \leq H(X|\mathcal{R}(BF_0))_a =H(X|BF_0)_b  \label{smaller entropy}.
\end{align}
Finally, by inserting \cref{equal-entropy} and \cref{smaller entropy} into \cref{asympt-rate} we obtain
\begin{align}
r_b \leq r_a,
\end{align}
which concludes the proof.
\end{proof}

\section{A Typical Telescope Model}
\label{App:telescope}

In this Appendix, we look at the implication of the two-mode model we have in \cref{fig:scenarios}, and deduce that the telescope action can be modelled by a beam-splitter like operation, where only one output mode is accessible. The gist of the idea is as follows. Let us denote by $a_r$ the field operator that will be collected by the telescope, after proper focusing, at point $r$ on the outer surface $S$ of the receiver telescope. We then have $[a_r,a_{r'}^\dag] = \delta(r-r')$, and the corresponding annihilation operator for the collected optical mode, in a particular polarization, is given by
\begin{align}
a = \int_S{{\rm d}r g(r) a_r},
\label{eq:telescope}
\end{align}
where $\int_S{{\rm d}r |g(r)|^2} =1$, hence $[a,a^\dag] = 1$. Here, we have assumed that the collected light is coupled to a single-mode fiber. 

In principle, the operator $\mathcal{E}_T$, acting in input modes $B$ and $F_0$ should give us the same output relationship as in \cref{eq:telescope}. In reality, in addition to the bypass channel and Eve's channel, the telescope could capture other background modes as well. In the worst-case scenario, however, we can always assume that all these other modes are controlled by Eve, and she can decide whether leave them as they are, or control them, via its operator $\mathcal{E}$. The implication of this assumption is that we can assume $\mathcal{E}_T$ is a unitary map, which fully models the action of the telescope. In particular, the collected light from mode $F_0$ combined with the collected light from mode $B$ must fully recover the action modeled by \cref{eq:telescope}. That is, if we model the collected light for mode $F_0$ by
\begin{align}
a_F = \int_S{{\rm d}r f(r) a_r},
\end{align}
with $\int_S{{\rm d}r |f(r)|^2} =1$, and the collected light for mode $B$ by
\begin{align}
a_B = \int_S{{\rm d}r h(r) a_r},
\end{align}
with $\int_S{{\rm d}r |h(r)|^2} =1$, we should then have
$[a_F,a_B] = 0$, as they originate from different spatial modes, and
\begin{align}
a = \alpha a_F + \beta a_B,
\end{align}
to make sure the two modes fully model the collected light by the telescope. The choice of linear combination above matches what a typical telescope does to different impinging modes of light. The first condition implies that the weight functions $f$ and $h$ must satisfy the orthogonality condition $\int_S{{\rm d}r f(r)h^\ast(r)} =0$, whereas the second condition implies that 
\begin{align}
g(r) = \alpha f(r) + \beta h(r),
\end{align}
which results in 
\begin{align}
\alpha = \int_S{{\rm d}r g(r) f^\ast(r)}, \beta =  \int_S{{\rm d}r g(r) h^\ast(r)}.
\end{align}
In addition, given that $g$, $f$ and $h$ are normalized and the latter two are orthogonal, we have $|\alpha|^2 + |\beta|^2 = 1$, which results in the following relationship
\begin{align}
\label{eq:tel_final}
a = \sqrt{\eta_{\rm T}}a_B + \sqrt{1-\eta_{\rm T}}a_F,
\end{align}
where
\begin{align}
\eta_{\rm T} = \int_S{{\rm d}r g(r) h^\ast(r)} = 1 - \int_S{{\rm d}r g(r) f^\ast(r)}.
\end{align}
The expression in \cref{eq:tel_final} resembles one output of a beam splitter with transmissivity $\eta_{\rm T}$, as we have used in the main text.

\section{Covariance matrix calculations}
\label{app:CVsetup}

In this Appendix, we calculate the covariance matrix (CM) for the setting given in \cref{fig:cv_telmodel}. While this is a special channel configuration, with proper choices of parameters, it can be used to model several cases of interest to our work. For instance, by choosing $\eta_{\rm S}$ to be zero, we effectively remove the bypass channel, and the remaining setup would then correspond to an optimal attack by Eve in the extended Alice-Bob model so long as the values assigned to $\eta_{\rm AE}$, $\eta_{\rm E}$, and $\eta_{\rm T}$ are within zero and one. 

To calculate the CM between all parties involved, i.e., Alice, Bob, and Eve, we consider the entanglement-based picture in \cref{fig:cv_telmodel} and start with the CM corresponding to the TMSV state $|\psi_{AB}\rangle$ with variance $V$, given by
\begin{align}
\textbf{V}_{AB}= \left(\begin{array}{cc}
V \mathbbm{1} & c {\mathbb Z}   \\
c {\mathbb Z} & V \mathbbm{1}
\end{array}\right) ,
\end{align}
where $c=\sqrt{V^2-1}$. On one leg of this TMSV state, Alice performs a heterodyne measurement, while she sends the other beam toward Bob. On its way, the latter beam experiences some pure loss, modelled by $\eta_{\rm AE}$, which splits the signal into two beams. One undergoes Eve's attack, whereby it would interfere, at a beam splitter with transmissivity $\eta_{\rm E}$, with Eve's TMSV state $|\psi_{EE'}\rangle$ with variance $V_E$, and the following CM
\begin{align}
\textbf{V}_{EE'}= \left(\begin{array}{cc}
V_E \mathbbm{1} & c_E {\mathbb Z}   \\
c_E {\mathbb Z} & V_E \mathbbm{1}
\end{array}\right),
\end{align}
where $c_E=\sqrt{V_E^2-1}$. The other output of $\eta_{\rm AE}$ beam splitter undergoes  additional loss, which is modelled via the beam splitter with transmissivity $\eta_{\rm S}$. Eventually, the two beams reconcile at the last beam splitter with transmissivity $\eta_{\rm T}$. 

Using linear optics algebra, we have modeled the above beam splitter operations using relevant matrices to find the CM of the purified state between all modes, i.e., $ABEE'F_0F_1$. After tracing out modes $F_0$ and $F_1$, as they are assumed inaccessible to all parties, we obtain
\begin{align}
\textbf{V}_{ABEE'}= \left(\begin{array}{cccc}
\label{CM:ABE}
V \mathbbm{1} & C_{AB} {\mathbb Z} &  0 \mathbbm{1}   & C_{AE'} {\mathbb Z}   \\
C_{AB} {\mathbb Z}   & V_B \mathbbm{1} &  C_{BE}  {\mathbb Z}   & C_{BE'}   \mathbbm{1}  \\
0 \mathbbm{1}   & C_{BE}  {\mathbb Z}    &  V_E \mathbbm{1} & C_{EE'}  {\mathbb Z}   \\
C_{AE'} {\mathbb Z}  & C_{BE'}\mathbbm{1}  &  C_{EE'} {\mathbb Z}  & V_{E'}    \mathbbm{1} 
\end{array}\right),
\end{align}
where the first row and column correspond to mode $A$ and its covariance elements with other modes, the second to $B$, and the third and forth to $E$ and $E'$, respectively. In \cref{CM:ABE},
\begin{align}
\label{eq:CM-param}
C_{AB} = &  \sqrt{T_{\rm eq}} c    \nonumber \\ 
C_{AE'} = & -\sqrt{\eta_{\rm AE}(1-\eta_{\rm E})} c   \nonumber \\  
V_B = & T_{\rm eq} (V-1) + 1 + \xi_{\rm eq}^{\rm Rx} &            \nonumber \\ 
C_{BE} = & \sqrt{(1-\eta_{\rm E})\eta_{\rm T}} c_E  \nonumber \\ 
C_{BE'} = & \sqrt{\eta_{\rm E} (1-\eta_{\rm E})\eta_{\rm T}} \Big( -\big(\eta_{\rm AE}(V-1)+1\big) +V_E \Big)   \nonumber \\  
& - \sqrt{ \eta_{\rm AE}(1-\eta_{\rm AE}) (1-\eta_{\rm E}) \eta_{\rm S} (1-\eta_{\rm T}) }(V-1)  \nonumber \\  
C_{EE'} = &  \sqrt{\eta_{\rm E}} c_E    \nonumber \\  
V_{E'} = &  (1-\eta_{\rm E})[\eta_{\rm AE}(V-1)+1]+\eta_{\rm E} V_E           \nonumber \\  
\end{align}
where
\begin{align}
\label{eq:Teq}
T_{\rm eq}= \Big( \sqrt{\eta_{\rm AE} \eta_{\rm E} \eta_{\rm T} }+\sqrt{(1-\eta_{\rm AE})\eta_{\rm S} (1-\eta_{\rm T})}\Big)^2,
\end{align}
appearing in the coefficient of $C_{AB}$ entry, is the observed value of transmissivity in the link, and
\begin{align}
\label{eq:Xieq}
\xi_{\rm eq}^{\rm Rx}= T_{\rm eq} \xi = (1-\eta_E)\eta_T (V_E-1)
\end{align}
is effectively the observed value of excess noise at the receiver, with $\xi$ being its equivalent at the transmitter end. As one would expect, the excess noise is a function of Eve's variance $V_E$ and is simply the amount of noise that enters Bob's receiver via the two beam splitters on the path between Bob and Eve. Similarly, $\sqrt{T_{\rm eq}}$ in \cref{eq:Teq} is the sum of the amplitudes in the two pathways from Alice to Bob. Similar calculations show that, if instead of the pure-loss bypass channel, we assume a thermal-loss bypass channel with a noise variance $V_{\rm S}$, there would be an additional term for $ \xi_{\rm eq}^{\rm Rx}$, given by $(1-\eta_{\rm S})(1-\eta_{\rm T})(V_{\rm S} -1)$, which accounts for the noise coming from the bypass channel, with no change in $T_{\rm eq}$.

The above CM can be used to calculate the key rate in different scenarios. For any given observed value of $T_{\rm eq}\leq 1$ and $\xi \geq 0$,  we can search the $\eta_{\rm S}-\eta_{\rm T}$ space for the minimum guaranteed key rate. One could also account for other sources of trusted noise at the receiver, such as electronic noise, by adjusting the above parameters, but for the purpose of our discussion on CV-QKD in the restricted case, the above framework is sufficiently detailed.

\bibliography{references}

\end{document}